\documentclass{llncs}

\usepackage{amsmath, amssymb}
\usepackage{dsfont}
\usepackage[colorlinks=true,urlcolor=webblue,linkcolor=webgreen,filecolor=webblue,citecolor=webgreen,pdfpagemode=UseOutlines,pdfstartview=FitH,pdfpagelayout=OneColumn,bookmarks=true]{hyperref}
\usepackage{fullpage}
\usepackage{enumerate}
\usepackage{tikz}
\usepackage[numbers]{natbib}

\definecolor{webgreen}{rgb}{0,.5,0}
\definecolor{webblue}{rgb}{0,0,.5}

\newtheorem{defn}{Definition}

\newtheorem{thm}[defn]{Theorem}
\newtheorem{prop}[defn]{Proposition}
\newtheorem{cor}[defn]{Corollary}
\newtheorem{lem}[defn]{Lemma}
\newtheorem{rem}[defn]{Remark}
\newtheorem{ex}[defn]{Example}
\numberwithin{defn}{section}
\numberwithin{equation}{section}

\newcommand{\R}{\mathbb{R}}
\newcommand{\N}{\mathbb{N}}
\newcommand{\C}{\mathbb{C}}

\newcommand{\id}{\mathrm{id}}

\newcommand{\ket}[1]{| #1 \rangle}
\newcommand{\bra}[1]{\langle #1 |}
\newcommand{\braket}[2]{\left\langle #1 \mid #2 \right\rangle}
\newcommand{\ketbra}[2]{\left|#1\right\rangle\!\!\left\langle #2\right|}

\newcommand{\tr}{\mathrm{Tr}}

\newcommand{\proj}[1]{\ensuremath{|#1\rangle \langle #1|}}

\renewcommand{\rho}{\varrho}

\newcommand{\D}{\mathrm{d}}
\newcommand{\supp}{\mathrm{supp}}

\newcommand{\spa}{\mathrm{span}}

\newcommand{\Hi}{\mathcal{H}}

\newcommand{\hi}{\Hi}

\newcommand{\one}{\mathds 1}

\newcommand{\opr}{\mathcal B}

{\begin{framed}\begin{small}}
{\end{small}\end{framed}}

\newcommand{\SKQES}{\textsf{QES}}
\newcommand{\ITS}{\textsf{ITS}}
\newcommand{\IND}{\textsf{IND}}
\newcommand{\ITNM}{\textsf{NM}}
\newcommand{\ABWNM}{\textsf{ABW-NM}}
\newcommand{\ABW}{\textsf{ABW}}
\newcommand{\DNS}{\textsf{DNS}}
\newcommand{\GYZ}{\textsf{GYZ}}
\newcommand{\acc}{\textsf{acc}}
\newcommand{\rej}{\textsf{rej}}
\newcommand{\expref}[2]{\texorpdfstring{\hyperref[#2]{#1~\ref{#2}}}{#1~\ref{#2}}}

\newcommand{\Enc}{\ensuremath{\mathsf{Enc}}}
\newcommand{\Dec}{\ensuremath{\mathsf{Dec}}}

\newcommand{\poly}{\operatorname{poly}}
\newcommand{\ttag}{\Pi^{\operatorname{tag}}_t}

\makeatletter
\renewcommand*\l@author[2]{}
\renewcommand*\l@title[2]{}
\makeatletter

\makeindex

\title{Quantum non-malleability and authentication}
\author{} \institute{}
\date{\today}
\author{Gorjan Alagic \and Christian Majenz}
\institute{QMATH, Department of Mathematical Sciences\\ University of Copenhagen\\~\\ \email{galagic@gmail.com} \qquad \email{majenz@math.ku.dk}}

\begin{document}

\maketitle

\begin{abstract}
In encryption, non-malleability is a highly desirable property: it ensures that adversaries cannot manipulate the plaintext by acting on the ciphertext. In~\cite{ambainis2009nonmalleable}, Ambainis et al. gave a definition of non-malleability for the encryption of quantum data. In this work, we show that this definition is too weak, as it allows adversaries to ``inject'' plaintexts of their choice into the ciphertext. We give a new definition of quantum non-malleability which resolves this problem. Our definition is expressed in terms of entropic quantities, considers stronger adversaries, and does not assume secrecy. Rather, we prove that \emph{quantum non-malleability implies secrecy}; this is in stark contrast to the classical setting, where the two properties are completely independent.  For unitary schemes, our notion of non-malleability is equivalent to encryption with a two-design (and hence also to the definition of~\cite{ambainis2009nonmalleable}).

\hspace{10pt} Our techniques also yield new results regarding the closely-related task of quantum authentication. We show that ``total authentication'' (a notion recently proposed by Garg et al.~\cite{Garg2017}) can be satisfied with two-designs, a significant improvement over the eight-design construction of~\cite{Garg2017}. We also show that, under a mild adaptation of the rejection procedure, both total authentication and our notion of non-malleability yield quantum authentication as defined by Dupuis et al.~\cite{dupuis2012actively}. 
\end{abstract}

\section{Introduction}

\subsubsection{Background.}
In its most basic form, encryption ensures secrecy in the presence of eavesdroppers. Besides secrecy, another desirable property is \emph{non-malleability}, which guarantees that an active adversary cannot modify the plaintext by manipulating the ciphertext. In the classical setting, secrecy and non-malleability are independent: there are schemes which satisfy secrecy but are malleable, and schemes which are non-malleable but transmit the plaintext in the clear. If both secrecy and non-malleability is desired, then pairwise-independent permutations provide information-theoretically perfect (one-time) security~\cite{kawachi2011characterization}. In the computational security setting, non-malleability can be achieved by MACs, and ensures chosen-ciphertext security for authenticated encryption.

In the setting of quantum information, encryption is the task of transmitting quantum states over a completely insecure quantum channel. Information-theoretic secrecy for quantum encryption is well-understood. Non-malleability, on the other hand, has only been studied in one previous work, by Ambainis, Bouda and Winter~\cite{ambainis2009nonmalleable}. Their definition (which we will call \ABW-non-malleability, or \ABWNM) requires that the scheme satisfies secrecy, and that the ``effective channel'' $\Dec \circ \Lambda \circ \Enc$ of any adversary $\Lambda$ amounts to either the identity map or replacement by some fixed state. In the case of unitary schemes, $\ABWNM$ is equivalent to encrypting with a unitary two-design. Unitary two-designs are a natural quantum analogue of pairwise-independent permutations, and can be efficiently constructed in a number of ways (see, e.g., ~\cite{Brandao2016a,dankert2009exact}.)

While quantum non-malleability has only been considered by~\cite{ambainis2009nonmalleable}, the closely-related task of quantum authentication (where decryption is allowed to reject) has received significant attention (see, e.g.,~\cite{ABE10, barnum2002authentication, Broadbent2016, dupuis2012actively,Hayden2016,Portmann2017, Garg2017}.) The widely-adopted definition of Dupuis, Nielsen and Salvail asks that the averaged effective channel of any adversary is close to a map which does not touch the plaintext~\cite{dupuis2012actively}; we refer to this notion as \DNS-authentication. Recent work by Garg, Yuen and Zhandry~\cite{Garg2017} established another notion of quantum authentication, which they call ``total authentication.'' The notion of total authentication has two major differences from previous definitions: (i.) it asks for success with high probability over the choice of keys, rather than simply on average, and (ii.) it makes no demands whatsoever in the case that decryption rejects. We refer to this notion of quantum authentication as \GYZ-authentication. In~\cite{Garg2017}, it is shown that \GYZ-authentication can be satisfied with unitary eight-designs. In~\cite{Hayden2016,Portmann2017}, a different approach is taken. Here, instead of using specialized security definitions, certain schemes are proven to provide an authenticated channel in the universal composability \cite{Unruh2004,Ben-Or2004a} and abstract cryptography \cite{Maurer2011} frameworks, respectively.

\subsubsection{This work.}

In this work, we devise a new definition of non-malleability (denoted \ITNM) for quantum encryption, improving on $\ABWNM$ in a number of ways. First, our definition is expressed in terms of entropic quantities, which allows us to bring several quantum-information-theoretic techniques to bear (such as decoupling.) Second, we consider more powerful adversaries, which can possess side information about the plaintext. Third, we remove the possibility of a ``plaintext injection'' attack, whereby an adversary against an $\ABWNM$ scheme can send a plaintext of their choice to the receiver. Finally, our definition does not demand secrecy; instead, we show that \emph{quantum secrecy is a consequence of quantum non-malleability.} This is a significant departure from the classical case, and is analogous to the fact that quantum authentication implies secrecy~\cite{barnum2002authentication}.

The primary consequence of our work is twofold: first, encryption with unitary two-designs satisfies all of the above notions of quantum non-malleability; second, when equipped with blank ``tag'' qubits, the same scheme also satisfies all of the above notions of quantum authentication. A more detailed summary of the results is as follows. For schemes which have unitary encryption maps, we prove that $\ITNM$ is equivalent to encryption with unitary two-designs, and hence also to $\ABWNM$. For non-unitary schemes, we prove a characterization theorem for $\ITNM$ schemes that shows that \ITNM~implies \ABWNM, and provide a strong separation example between $\ITNM$ and $\ABWNM$ (the aforementioned plaintext injection attack). In the case of \GYZ~authentication, we prove that two-designs (with tags) are sufficient, a significant improvement over the state-of-the-art, which requires eight-designs~\cite{Garg2017}. Moreover, the simulation of adversaries in this proof is efficient, in the sense of Broadbent and Wainewright~\cite{Broadbent2016}. Finally, we show that \GYZ authentication implies \DNS-authentication, and that equipping an arbitrary $\ITNM$ scheme with tags yields \DNS-authentication.

We remark that, after the initial version of our results was completed and submitted, an independent work of C. Portmann on quantum authentication appeared~\cite{Portmann2017}; it gives an alternative proof that \GYZ-authentication can be satisfied by the 2-design scheme. 

\subsection{Summary of contributions}

A more technical summary of our contributions follows. We remark here that all of our results concern information-theoretic security notions for symmetric-key encryption of quantum data. 

\subsubsection{Quantum non-malleability.}

Our first set of results is concerned with non-malleability. The results hold in both the perfect setting (\expref{Section}{sec:zero-error}) and the approximate setting (\expref{Section}{sec:appr}). 
\begin{enumerate}

	\item \textbf{New definition of non-malleability.} We give a new definition of quantum non-malleability (\ITNM), in terms of the information gain of an adversary's \emph{effective attack} on the plaintext. The relevant quantum registers are: plaintext $A$, ciphertext $C$, user's reference $R$, and adversary's side information $B$.
\begin{defn}[$\ITNM$, informal]
A scheme is non-malleable (\ITNM) if for any $\rho_{ABR}$ and any attack $\Lambda_{CB \to C\tilde B}$, the effective attack $\tilde \Lambda_{AB \to A\tilde B}$ satisfies
$$
	I(AR:\tilde B)_{\tilde\Lambda(\rho)} \leq I(AR:B)_\rho + h(p_{=}(\Lambda,\rho)).
$$
\end{defn}
The binary entropy term is necessary because adversaries can always simply record whether they disturbed the ciphertext (see \expref{Definition}{def:qNM}).

		\item \textbf{Results on non-malleability.} Our first result is that $\ITNM$ implies secrecy. 
	\begin{thm}[informal]
		For quantum encryption, non-malleability implies secrecy.
	\end{thm}
We also show that $\ITNM$ implies \ABWNM, and give a separation scheme which is secure under \ABWNM~but insecure under \ITNM. This scheme is in fact susceptible to a powerful attack, whereby a simple adversary can replace the output of decryption with a plaintext of the adversary's choice.  On the other hand, if we restrict our attention to schemes where the encryption maps are unitary, then we are able to show the following.
	\begin{thm}[informal]\label{thm:intro-2}
		Let $\Pi$ be a scheme such that encryption $E_k$ is unitary for all keys $k$. Then $\Pi$ is $\ITNM$ if and only if $\{E_k\}_k$ is a two-design.
	\end{thm}
	By the results of~\cite{ambainis2009nonmalleable}, we conclude that \ITNM~and \ABWNM~are in fact equivalent for unitary schemes. Finally, we can also characterize $\ITNM$ schemes in the general (i.e., not necessarily unitary) case, as follows.

\begin{thm}[informal]
A scheme is $\ITNM$ if and only if, for any $\Lambda_{CB\to C\tilde B}$, there exist maps $\Lambda'_{B\to\tilde B}$, $\Lambda''_{B\to\tilde B}$ such that the effective attack $\tilde \Lambda_{AB\to A\tilde B}$ has the form
$$
\tilde \Lambda =\id_A\otimes \Lambda'+\frac{1}{|C|^2-1}\left(|C|\left\langle D_K(\one_C)\right\rangle-\id\right)_A\otimes \Lambda''\,.
$$
\end{thm}

		\item \textbf{Authentication from non-malleability.} Our final result in the setting of non-malleability shows that, by adding a ``tag'' space to the plaintext (as seen, e.g., in the Clifford scheme~\cite{ABE10}), we can turn a non-malleable scheme into an authentication scheme as defined in~\cite{dupuis2012actively}. More precisely, given an encryption scheme $\Pi = \{E_k\}$, we define $\ttag$ to be a new scheme whose encryption is $\rho \mapsto E_k  \bigl( \rho_A \otimes \ketbra{0}{0}^{\otimes t}_B\bigr) E_k^\dagger$, and whose decryption rejects unless $B$ measures to $\ket{0^t}$.
	\begin{thm}[informal]
		Let $\Pi = \{E_k\}$ be an encryption scheme. If $\Pi$ is $\ITNM$, then $\ttag$ is $2^{2-t}$-DNS-authenticating.
	\end{thm}
\end{enumerate}

\subsubsection{Quantum authentication.}
Our results on quantum authentication are summarized as follows. We note that, strictly speaking, our definitions of authentication deviate slightly from the original versions~\cite{dupuis2012actively, Garg2017}, in that decryption outputs a reject symbol in place of the plaintext (rather than setting an auxiliary bit to ``reject.'') This adaptation is convenient for reasons we will return to later.

\begin{enumerate}

\item \textbf{\GYZ~implies \DNS.}  First, we show that \GYZ-authentication implies \DNS-authentication. We remark that this is not trivial: on one hand, \GYZ~strengthens \DNS~by requiring high probability of success (rather than success on-average); on the other hand, in the reject case \GYZ~requires nothing while \DNS~makes rather stringent demands. Nonetheless, we show the following.
\begin{thm}[informal]
Let $\Pi$ be an encryption scheme. If $\Pi$ is $\varepsilon$-\GYZ-authenticating, then it is also $O(\sqrt\varepsilon)$-\DNS-authenticating.
\end{thm}

\item \textbf{\GYZ~is achievable with 2-designs.}
Next, we show that \GYZ-authentication can be satisfied by a scheme which ``tags'' plaintexts as before, and encrypts with a unitary 2-design. This is a significant improvement over the analysis of~\cite{Garg2017}, which required eight-designs for the same construction.

\begin{thm}[informal]\label{thm:intro-3}
		Let $\Pi = \{ E_k \}_k$ be a $2^{-t}$-approximate 2-design scheme. Then $\ttag$ is $2^{-\Omega(t)}$-\GYZ-authenticating.
\end{thm}

\item \textbf{$\GYZ$ authentication from non-malleability.}
Finally, we record a straightforward consequence of \expref{Theorem}{thm:intro-2} and \expref{Theorem}{thm:intro-3}: tagging a unitary non-malleable scheme results in a \GYZ-authenticating scheme.

\begin{cor}[informal]
There exists a constant $r > 0$ such that the following holds. If $\Pi$ is a unitary $\Omega(2^{-r n})$-$\ITNM$ scheme for $n$-qubit messages, and $t = \poly(n)$, then $\ttag$ is $2^{-\Omega(\poly(n))}$-GYZ-authenticating. 
\end{cor}

We remark that a sufficiently strong $\ITNM$ scheme for the above can be constructed via the $\epsilon$-approximate version of \expref{Theorem}{thm:intro-2} (see \expref{Theorem}{thm:eps-USKQES-NM-2design} and \expref{Remark}{rem:efficient} below.)
\end{enumerate}

The remainder of the paper is structured as follows. In \expref{Section}{sec:prelims}, we review some basic facts regarding quantum states, registers, and channels, and recall several useful facts about unitary designs. In \expref{Section}{sec:zero-error}, we consider the exact setting, beginning with perfect secrecy and then continuing to perfect non-malleability (\ITNM) and the relevant new results; we also discuss the relationship to \ABWNM~in detail. 
While the exact setting is relatively simple and conceptually clean, it is not relevant in practical settings. We thus continue in \expref{Section}{sec:appr} by developing the approximate setting, again beginning with secrecy and then continuing to approximate non-malleability. We end with the new results on quantum authentication, in \expref{Section}{sec:authentication}.

\section{Preliminaries}\label{sec:prelims}

\subsection{Quantum states, registers, and channels.}

We assume basic familiarity with the formalism of quantum states, operators, and channels. We denote quantum registers (i.e., systems and their subsystems) with capital Latin letters, e.g., $A, B, C$. The Hilbert space corresponding to system $A$ is denoted by $\hi_A$. For a register $A$, we denote the dimension of $\hi_A$ by $|A|$. We emphasize that, in this work, all Hilbert spaces will be finite-dimensional.

The space operators on $\hi_A$ is denoted $\opr(\hi_A)$. We say that a quantum state is classical if it is diagonal in the standard (i.e., computational) basis. We denote the adjoint of an operator $X \in \opr(\hi)$ by $X^\dagger$ and its transose with respect to the computational basis by $X^T$. Where necessary, we will write a quantum state $\rho \in \opr(\hi_A \otimes \hi_B \otimes \hi_C)$ as  $\rho_{ABC}$ to emphasize that the state is a multipartite state over registers $A$, $B$, and $C$. When such a state has already been defined, we will write reduced states by omitting the traced-out registers, e.g., $\rho_A := \tr_{BC} [\rho_{ABC}]$. We single out some special states which will appear frequently. Fix two systems $S, S'$ with $|S| = |S'|$. We let 
$$
\ket{\phi^+}_{SS'} = |S|^{-1/2}\sum_i\ket{ii}_{SS'}
\qquad \text{and} \qquad
\phi^+_{SS'} = \proj{\phi^+}_{SS'}
$$
denote the maximally entangled state on the bipartite system $SS'$ (expressed as a pure state on the left, and as a density operator on the right.) Furthermore, we let $\Pi^-_{SS'}=\mathds 1_{SS'}-\phi^+_{SS'}$ and $\tau^-_{SS'}=\Pi^-_{SS'}/(|S|^2-1)$. We also set $\tau_S=\mathds 1_S/|S|$ to be the maximally mixed state on $S$.

We denote the von Neumann entropy of a state $\rho_A$ by $H(A)_\rho$, and the joint entropy of $\rho_{AB}$ by $H(AB)_\rho$. We recall that the quantum mutual information of $\rho_{AB}$ is defined by
$$
I(A:B)_\rho := H(A)_\rho + H(B)_\rho - H(AB)_\rho\,.
$$
The quantum conditional mutual information of $\rho_{ABC}$ is defined by
$$
I(A:B|C)_\rho := H(AC)_\rho + H(BC)_\rho - H(ABC)_\rho - H(C)_\rho\,.
$$
These quantities are nonnegative \cite{lieb1973fundamental} and satisfy a chain rule:
$$
I(A:BC|D)_\rho=I(A:B|D)_\rho+I(A:C|BD)_\rho\,.
$$
We remark that the above also holds for trivial $D$. Together with the Stinespring dilation theorem \cite{stinespring1955positive}, nonnegativity \cite{lieb1973fundamental} and the chain rule imply the data processing inequality
$$
I(A:\tilde B|C)_{\Lambda(\rho)}\le I(A:B|C)_\rho\,,
$$
when $\Lambda$ is a CPTP (completely-positive, trace-preserving) map from $\opr(\hi_B)$ to $\opr(\hi_{\tilde B})$. An important special case is where $B=B_1B_2$ and $\Lambda=\tr_{B_2}$ discards the contents of $B_2$.

We will refer to valid transformations between quantum states as channels, or CPTP maps. We will sometimes also consider trace-non-increasing completely-positive (CP) maps. When necessary, we will emphasize the input and output spaces of a map 
$$
\Lambda : \opr(\hi_A\otimes \hi_B) \rightarrow \opr(\hi_C)
$$
by writing $\Lambda_{AB \to C}$. We denote the identity channel on, e.g.,  register $A$ by $\mathrm{id}_{A\to A}$ (or simply $\mathrm{id}_A$) and the channel from register $A$ to $A'$ with constant output $\sigma_{A'}$ by $\langle\sigma\rangle_{A\to A'}$. When composing operators on many registers, and if the context allows, we will elide tensor products with the identity operator. So, for example, with $\Lambda$ as above we may write $\tau_{CD} = \Lambda\, \rho_{ABD}$ in place of $\tau_{CD} = (\Lambda \otimes \id_D) \rho_{ABD}$.

A standard tool for analyzing operators on matrix spaces is the Choi-Jamio\l kowski (CJ) isomorphism \cite{choi1975completely,jamiolkowski1972linear}. Let $\Xi_{A\to B}: \opr(\hi_A)\to \opr(\hi_B)$ be a linear operator. Then its CJ matrix is defined as
\begin{equation}
	\left(\eta_{\Xi}\right)_{BA'}=\Lambda_{A\to B}(\phi^+_{AA'}).
\end{equation}
The linear operator mapping $\Xi$ to $\eta_{\Xi}$ is an isomorphism of vector spaces and $\eta_\Xi$ is positive semidefinite iff $\Xi$ is CP. Moreover $\Xi_{A\to B}$ is TP iff $\left(\eta_{\Xi}\right)_{A'}=\tau_A$. The inverse of the CJ isomorphism is given by the equation
	\begin{equation}
		\Xi_{A\to B}(X_A)=|A|\tr_{A'}\left[X_{A'}^T\left(\eta_{\Xi}\right)_{BA'}\right].
	\end{equation}

 An operator which will appear often is the swap: $F : \ket{i}\otimes\ket j \mapsto \ket j\otimes\ket i$. We will frequently refer to the following fact; a proof appears in~\cite{dupuis2014one}. 

\begin{lem}[Swap trick]\label{lem:swap-trick}
For any matrices $A, B$ we have $\tr [AB]=\tr [F  A\otimes B]$. 
\end{lem}

We denote the trace norm by $\|\cdot\|_1$, the infinity (or operator) norm by $\|\cdot\|_\infty$, and the diamond norm\footnote{This is the dual of the completely bounded norm: $\|\Lambda_{A\to B}\|_\diamond=\max_{\rho_{AA'}}\|\Lambda_{A\to B}\otimes\id_{A'}(\rho_{AA'})\|_1$, where the max is taken over all pure quantum states $\rho_{AA'}$ and $\hi_A\cong\hi_{A'}$.} by $\|\cdot\|_\diamond$. A basic inequality we will make use of is the H{\"o}lder inequality for operators, which we take to be
\begin{equation}\label{eq:holder}
\tr[XY] \le \|XY\|_1\le\|X\|_1\|Y\|_\infty\,,
\end{equation}
for any two operators $X$ and $Y$. 

\subsection{Unitary designs.}

We now recall the definition of unitary $t$-design, and some relevant variants. We begin by considering three different types of ``twirls.''
\begin{enumerate}
\item For a finite subset $\mathrm D\subset \mathrm{U}(\hi)$ of the unitary group on some finite dimensional Hilbert space $\hi$, let
\begin{equation}
\mathcal T^{(t)}_{\mathrm D}(X)=\frac{1}{|\mathrm D|}\sum_{U\in\mathrm D}U^{\otimes t}X\left(U^\dagger\right)^{\otimes t}
\end{equation}
be the associated $t$-twirling channel. If we take the entire unitary group (rather than just a finite subset), then we get the Haar $t$-twirling channel
\begin{equation}\label{eq:haar-twirl}
\mathcal T^{(t)}_\mathsf{Haar}(X)=\int U^{\otimes t}X\left(U^\dagger\right)^{\otimes t} \D U.
\end{equation}

The case $t = 2$ is characterized in \expref{Lemma}{lem:Usquared} in \expref{Appendix}{appendix:technical}.

\item We define the $U$-$\overline{U}$ twirl with respect to finite $\mathrm D\subset \mathrm{U}(\hi)$ by
\begin{equation}\label{eq:bar-twirl}
\overline{\mathcal{T}}_{\mathrm D}(X)
=\frac{1}{|\mathrm D|}\sum_{U\in\mathrm D}\left(U\otimes\overline U \right)X\left(U\otimes \overline{U}\right)^\dagger.
\end{equation}
The analogous $U$-$\overline{U}$ Haar twirling channel is denoted by $\overline {\mathcal T}_{\mathsf{Haar}}$. 

\item The third and final notion is called a channel twirl; this is defined in terms of $U$-$\overline{U}$-twirling, as follows. Given a channel $\Lambda$, let $\eta_\Lambda$ be the CJ state of $\Lambda$. The channel twirl $\mathcal T_{\mathrm{D}}^{ch}(\Lambda)$ of $\Lambda$ is defined to be the channel whose CJ state is $\overline{\mathcal{T}}_{\mathrm D}(\eta_\Lambda)$.
\end{enumerate}

Next, we define three notions of designs, corresponding to the three types of twirl defined above.

\begin{defn}\label{def:designs}
	Let $\mathrm D\subset \mathrm{U}(\hi)$ be a finite set. We define the following.
	\begin{itemize}
		\item If $\bigl\|\mathcal T^{(t)}_{\mathrm D}-\mathcal T^{(t)}_\mathsf{Haar}\bigr\|_\diamond\le\delta$ holds, then $\mathrm{D}$ is a $\delta$-approximate $t$-design.
		\item If $\bigl\|\overline{\mathcal T}_{\mathrm D}-\overline{\mathcal T}_\mathsf{Haar}\bigr\|_\diamond\le\delta$ holds, then $\mathrm{D}$ is a $\delta$-approximate $U$-$\overline{U}$-twirl design.
		\item If $\left\|\mathcal T^{ch}_{\mathrm D}(\Lambda)-T^{ch}_\mathsf{Haar}(\Lambda)\right\|_\diamond\le\delta$ holds for all CPTP maps $\Lambda$, then $\mathrm{D}$ is a $\delta$-approximate channel-twirl design.
	\end{itemize}
\end{defn} 
For all three of the above, the case $\delta = 0$ is called an ``exact design'' (or simply ``design''.) All three notions of design are equivalent in the exact case. In the approximate case they are still connected, but there are some nontrivial costs in the approximation quality (See \cite{low2010pseudo}, Lemma 2.2.14, and \expref{Lemma}{lem:channel-twirl-uubar-twirl} in \expref{Appendix}{appendix:technical}).

It is well-known that $\varepsilon$-approximate $t$-designs on $n$ qubits can be generated by random quantum circuits of size polynomial in $n, t$ and $\log(1/\varepsilon)$~\cite{Brandao2016a}. In particular, the size of these circuits is polynomial even for exponentially-small choices of $\varepsilon$. We emphasize this observation as follows.
\begin{rem}\label{rem:efficient}
Fix a polynomial $t$ in $n$. Then, for any $\varepsilon > 0$, a random $n$-qubit quantum circuit consisting of $\poly(n, \log(1/\varepsilon))$ gates (from a universal set) satisfies every notion of $\epsilon$-approximate $t$-design in \expref{Definition}{def:designs}. 
\end{rem}

For exact designs, we point out two important constructions. First, the prototypical example of a unitary one-design on $n$ qubits is the $n$-qubit Pauli group. For exact unitary two-designs, the standard example is the Clifford group, which is the normalizer of the $n$-qubit Pauli group. Alternatively, the Clifford group is generated by circuits from the gate set $\{H, P, \text{CNOT}\}$. It is well-known that one can efficiently generate exact unitary two-designs on $n$-qubits by building appropriate circuits from this gate set, using $O(n^2)$ random bits~\cite{Aaronson2004, dankert2009exact}.


\section{The zero-error setting}\label{sec:zero-error}

We begin with the zero-error. In the case of secrecy, zero-error means that schemes cannot leak any information whatsoever. In the case of non-malleability, zero-error means that the adversary cannot increase their correlations with the secret by even an infinitesimal amount (except by trivial means; see below.)

\subsection{Perfect secrecy}

We begin with a definition of symmetric-key quantum encryption. Our formulation treats rejection during decryption in a slightly different manner from previous literature.

\begin{defn}[Encryption scheme]\label{def:SKQES}
A symmetric-key quantum encryption scheme (\SKQES) is a triple $(\tau_K,E,D)$ consisting of a classical state $\tau_K\in \opr(\hi_K)$ and a pair of channels
\begin{align*}
E &: \opr(\hi_A\otimes \hi_K) \longrightarrow \opr(\hi_C\otimes \hi_K)\\
D &: \opr(\hi_C\otimes \hi_K)  \longrightarrow \opr(\left(\hi_A\oplus\C\ket{\bot}\right)\otimes \hi_K)
\end{align*}
satisfying $[D\circ E](\cdot\otimes \proj k)= (\mathrm{id}_{A} \oplus 0_\bot) \otimes \proj k$ for all $k$. 
\end{defn}
The Hilbert spaces $\hi_A$, $\hi_C$ and $\hi_K$ are implicitly defined by the triple $(\tau_K, E, D)$. The state $\ket{\bot}$ is an error flag that allows the decryption map to report an error. For notational convenience when dealing with these schemes, we set
\begin{alignat*}{2}
E_k &= E(\cdot\otimes \proj k)
\qquad \qquad 
&E_K &= \tr_KE(\cdot\otimes \tau_K)\\
D_k &= D(\cdot\otimes \proj k)
&D_K &= \tr_KD(\cdot\otimes \tau_K)\,.
\end{alignat*}
We will often slightly abuse notation by referring to decryption maps $D_k$ as maps from $C$ to $A$; in fact, the output space of $D_k$ is really the slightly larger space $\bar A := A \oplus \C\ket{\bot}$. 

It is natural to define secrecy in the quantum world in terms of quantum mutual information. However, instead of asking for the ciphertext to be uncorrelated with the plaintext as in the classical case, we ask for the ciphertext to be uncorrelated from any reference system. 

\begin{defn}[Perfect secrecy]\label{def:ITS}
A \SKQES~$(\tau_K, E, D)$ satisfies information - theoretic secrecy (\ITS) if, for any Hilbert space $\hi_B$ and any $\rho_{AB}\in \opr(\hi_A\otimes \hi_B)$, setting $\sigma_{CBK}=E(\rho_{AB}\otimes \tau_K)$ implies $I(C:B)_\sigma=0.$
\end{defn}

We note that, for perfect \ITS, adding side information is unnecessary: the definition already implies that the ciphertext is in product with \emph{any} other system. In particular, if the adversary has some auxiliary system $E$ in their possession, then $I(B:CE)_\sigma=I(B:E)_\sigma$. Several definitions of secrecy for symmetric-key quantum encryption have appeared in the literature, but the above formulation appears to be new. It can be shown that $\ITS$ is equivalent to perfect indistinguishability of ciphertexts (\IND). The latter notion is a special case of an early indistinguishability-based definition of Ambainis et al.~\cite{ambainis2000private}. 

In many situations it makes sense to restrict ourselves to \SKQES~that have identical plaintext and ciphertext spaces; due to correctness, this is equivalent to unitarity.
\begin{defn}[Unitary scheme]
	A \SKQES~$(\tau_K,E,D)$ is called unitary if the encryption and decryption maps are controlled unitaries, i.e., if there exists 
	$V = \sum_k U^{(k)}_A\otimes\proj k_K$ such that $E(X)=VXV^\dagger$.
\end{defn}
It is straightforward to prove that, for unitary schemes, \ITS~is equivalent to the statement that the encryption maps $\{E_k\}$ form a unitary 1-design. Note that unitarity of $E_k$ and correctness imply unitarity of $D_k$.

\subsection{A new notion of non-malleability}

\subsubsection{Definition.}
We consider a scenario involving a user Alice and an adversary Mallory. The scenario begins with Mallory preparing a tripartite state $\rho_{ABR}$ over three registers:  the plaintext $A$, the reference $R$, and the side-information $B$. The registers $A$ and $R$ are given to Alice, while Mallory keeps $B$. Alice then encrypts $A$ into a ciphertext $C$ and then transmits (or stores) it in the open. Mallory now applies an attack map
$$
\Lambda : \opr(\hi_C\otimes\hi_B) \to \opr(\hi_C\otimes\hi_{\tilde B})\,.
$$
Mallory keeps the (transformed) side-information $\tilde B$ and returns $C$ to Alice. Finally, Alice decrypts $C$ back to $A$, and the scenario ends. 
\begin{figure}[h]\label{fig:qnm}
\begin{center}
\includegraphics[width=0.5\textwidth]{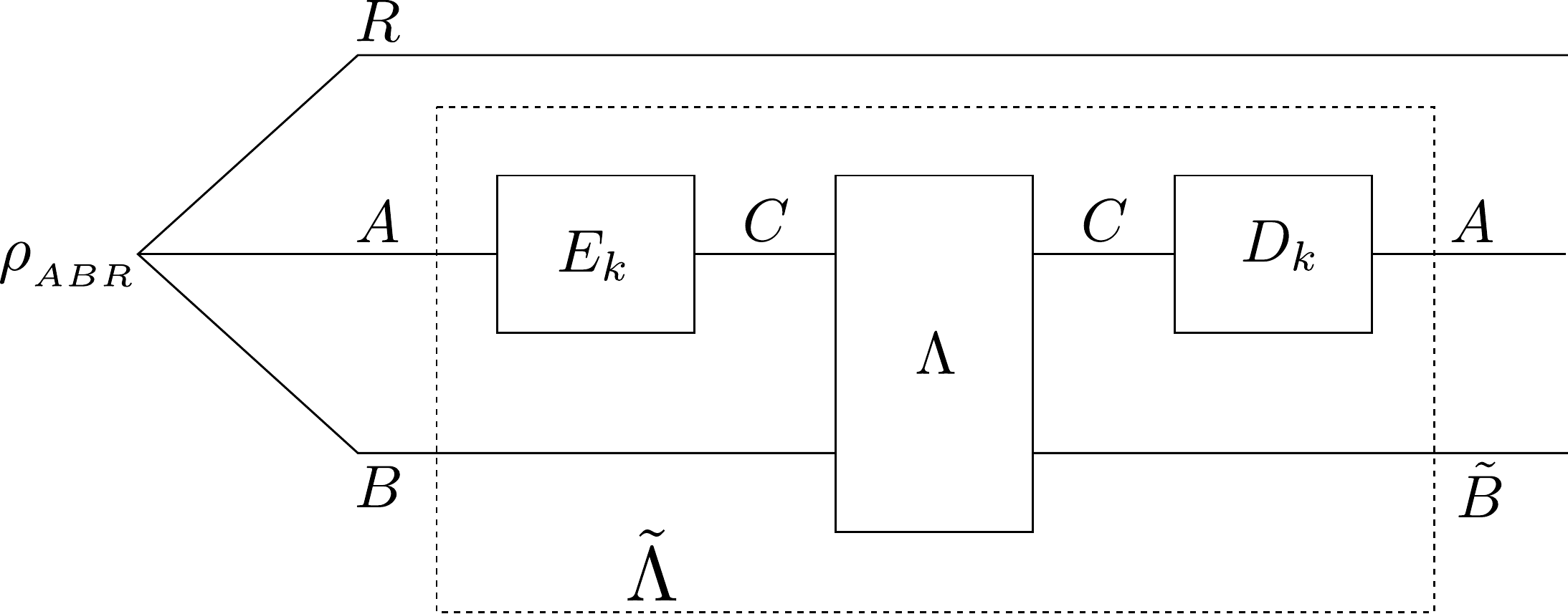}
\end{center}
\caption{The quantum non-malleability scenario.}
\end{figure}
We are now interested in measuring the extent to which Mallory was able to increase her correlations with Alice's systems $A$ and $R$. This can be understood by analyzing the mutual information $I(AR:\tilde B)_{\tilde\Lambda(\rho)}$ where $\tilde \Lambda_{AB \to A\tilde B}$ is the \emph{effective channel} corresponding to Mallory's attack:
\begin{equation}\label{eq:effective-channel}
\tilde\Lambda= \tr_K (D\circ\Lambda\circ E)((\cdot)\otimes \tau_K)\,.
\end{equation}
We point out one way in which Mallory can always increase these correlations, regardless of the structure of the encryption scheme. First, she flips a coin $b$, and records the outcome in $B$. If $b=1$, she replaces the contents of $C$ with some fixed state $\sigma_C$, and otherwise she leaves $C$ untouched. One then sees that Mallory's correlations have increased by $h(p_{=}(\Lambda,\rho))$, where $h$ denotes binary entropy and $p_{=}$ is a defined as follows.
\begin{equation}\label{eq:p-equals}
p_{=}(\Lambda, \rho) = 
\tr\left[ (\phi^+_{CC'}\otimes \mathds 1_{\tilde B}) \Lambda (\phi^+_{CC'} \otimes \rho_B)\right]\,.
\end{equation}
This quantity is the inner product between the identity map and the map $\Lambda((\,\cdot\,) \otimes \rho_B)$, expressed in terms of CJ states. Intuitively, it measures the probability with which Mallory chooses to apply the identity map; taking the binary entropy then gives us the information gain resulting from recording this choice.

We are now ready to define information-theoretic quantum non-malleability. Stated informally, a scheme is non-malleable if Mallory can only implement the attacks described above. 

\begin{defn}[Non-malleability]\label{def:qNM}
	A \SKQES~$(\tau_K,E, D)$ is non-malleable (\ITNM) if for any state $\rho_{ABR}$ and any CPTP map
$\Lambda_{CB \to C{\tilde B}}$, we have
	\begin{equation}\label{eq:ITNM-condition}
		I(AR:\tilde B)_{\tilde\Lambda(\rho)} \leq I(AR:B)_\rho + h(p_{=}(\Lambda,\rho)).
	\end{equation}
\end{defn}

One might justifiably wonder if the term $h(p_{=}(\Lambda, \rho))$ is too generous to the adversary. However, as we showed above, every scheme is vulnerable to an attack which gains this amount of information. This term also appears (somewhat disguised) in the classical setting. In fact, if a classical encryption scheme satisfies \expref{Definition}{def:qNM} against classical adversaries, then it also satisfies classical information-theoretic non-malleability as defined in \cite{kawachi2011characterization}. 
Finally, as we will show in later sections, \expref{Definition}{def:qNM} implies $\ABWNM$ (see \expref{Definition}{def:ABWNM}), and schemes satisfying \expref{Definition}{def:qNM} are sufficient for building quantum authentication under the strongest known definitions.

\subsubsection{Non-malleability implies secrecy.}

In the classical case, non-malleability is independent from secrecy: the one-time pad is secret but malleable, and non-malleability is unaffected by appending the plaintext to each ciphertext. In the quantum case, on the other hand, we can show that $\ITNM$ implies secrecy. This is analogous to the fact that ``quantum authentication implies encryption''~\cite{barnum2002authentication}. The intuition is straightforward: (i.) one can only make use of one copy of the plaintext due to no-cloning, and (ii.) if the adversary can distinguish between two computational-basis states (e.g., $\ket{0}$ and $\ket{1}$) then they can also apply an undetectable Fourier-basis operation (e.g., mapping $\ket{+}$ to $\ket{-}$). The technical statement and proof follows.

\begin{prop}\label{thm:ITNMtoITS}
	Let $(\tau_K,E, D)$ be an \ITNM~\SKQES. Then $(\tau_K,E, D)$ is \ITS.
\end{prop}
\begin{proof}
Let $B$, $\rho_{AB}$, and $\sigma_{CBK} = E(\rho_{AB} \otimes \tau_K)$ be as in the definition of \ITS~(\expref{Definition}{def:ITS}). We first rename $B$ to $R$. We then consider the non-malleability property in the following special-case scenario. The initial side-information register is empty, the final side-information register $\tilde B$ satisfies $\hi_{\tilde B} \cong \hi_C$, and the adversary map $\Lambda_{C\to C\tilde B}$ is defined as follows. Note that the ``ciphertext-extraction'' map $\Theta_{C\to C\tilde B}=\mathrm{id}_{C\to \tilde B}(\cdot)\otimes \tau_C$ has CJ state $\eta^{\Theta}_{CC'\tilde B}=\phi^+_{C'\tilde B}\otimes \tau_C$. We choose $\Lambda$ so that its CJ state satisfies
	\begin{equation}
	\eta^{\Lambda}_{CC'\tilde B}=\frac{d^2}{d^2-1}\Pi_{CC'}^- \,\eta^{\Theta}_{CC'\tilde B}\,\Pi_{CC'}^-\,.
	\end{equation}
Applying the above projection to the CJ state of $\Theta$ ensures that $\Lambda$ will have $p_=({\Lambda})=0$ (note: $p_=(\Theta) > 0$.)
		
Direct calculation of the $C' \tilde B$ marginal of the CJ state of $\Lambda$ yields
	\begin{equation}
	\eta^{\Lambda}_{C'\tilde B}=\frac{d^2-2}{d^2-1}\phi^+_{C'\tilde B}+\frac{1}{d^2-1}\tau_{C'}\otimes\tau_{\tilde B}.
	\end{equation}
This implies that the output $\sigma_{AR\tilde B}=\tilde\Lambda_{A\to A\tilde B}(\rho_{AB})$ of the effective channel $\tilde \Lambda$ will satisfy
	\begin{equation}\label{eq:eq1}
	\sigma_{\tilde BR}=\frac{d^2-2}{d^2-1}\gamma_{\tilde BR}+\frac{1}{d^2-1}\tau_{\tilde B}\otimes\rho_R,
	\end{equation}
	where $\gamma_{CR}=(E_K)_{A\to C}(\rho_{AR})$ and we used the fact that $\hi_{\tilde B} \cong \hi_C$. By non-malleability, we have
	\begin{equation}\label{eq:eq2}
	I(\tilde B:R)_{\sigma}+I(\tilde B:A|R)_{\sigma}=I(\tilde B:AR)_{\sigma}=0.
	\end{equation}
	In particular, $I(\tilde B:R)_{\sigma}=0$ and thus $\sigma_{\tilde BR}=\sigma_{\tilde B}\otimes \rho_R.$	
	It follows by Equation \eqref{eq:eq1} that
	\begin{equation}
	\gamma_{\tilde BR}=\frac{d^2-1}{d^2-2}\left(\sigma_{\tilde B}-\frac{1}{d^2-1}\tau_{\tilde B}\right)\otimes\rho_R,
	\end{equation}
	i.e., $\gamma_{\tilde BR}$ is a product state. This is precisely the definition of information-theoretic secrecy.
	\qed
\end{proof}

\subsubsection{Characterization of non-malleable schemes.}\label{sec:effective-char}

Next, we provide a characterization of non-malleable schemes. First, we show that unitary schemes are equivalent to encryption with a unitary 2-design.

\begin{thm}\label{thm:USKQES-NM-2design}
A unitary \SKQES~$(\tau_K, E, D)$ is \ITNM~if and only if $\{E_k\}_{k\in K}$ is a unitary 2-design.
\end{thm}
This fact is particularly intuitive when the 2-design is the Clifford group, a well-known exact 2-design. In that case, a Pauli operator acting on only one ciphertext qubit will be ``propagated'' (by the encryption circuit) to a completely random Pauli on all plaintext qubits. The plaintext is then maximally mixed, and the adversary gains no information. The Clifford group thus yields a perfectly non-malleable (and perfectly secret) encryption scheme using $O(n^2)$ bits of key~\cite{Aaronson2004}.

It will be convenient to prove \expref{Theorem}{thm:USKQES-NM-2design} as a consequence of our general characterization theorem, which is as follows. 

\begin{thm}\label{thm:effective-char}
	Let $(\tau, E, D)$ be a \SKQES. Then $(\tau, E, D)$ is $\ITNM$ if and only if, for any attack $\Lambda_{CB\to C\tilde B}$, the effective map $\tilde \Lambda_{AB\to A\tilde B}$ has the form
		\begin{equation}\label{eq:effective-char}
		\tilde \Lambda =\id_A\otimes \Lambda'_{B\to\tilde B}+\frac{1}{|C|^2-1}\left(|C|^2\left\langle D_K(\tau)\right\rangle-\id\right)_A\otimes \Lambda''_{B\to\tilde B}
		\end{equation}
		where $\Lambda' =\tr_{CC'}[\phi^+_{CC'}\Lambda(\phi^+_{CC'}\otimes (\cdot))]$ and $\Lambda'' =\tr_{CC'}[\Pi^-_{CC'}\Lambda(\phi^+_{CC'}\otimes (\cdot))].$

\end{thm}

We remark that the forward direction holds even if $(\tau, E, D)$ only fulfills the $\ITNM$ condition (Equation \eqref{eq:ITNM-condition}) against adversaries with empty side-information $B$. The proof of \expref{Theorem}{thm:effective-char} (with this strengthening) is sketched below. The full proof is somewhat technical and can be found in \expref{Appendix}{app:proofs}. More precisely, we prove the stronger \expref{Theorem}{thm:eps-effective-char-app}, which implies the above by setting $\varepsilon=0$.

\textit{Proof sketch.}
The first implication, i.e. $\ITNM$ implies Equation \eqref{eq:effective-char}, is best proven in the Choi-Jamoi\l kowski picture. Here, any $\SKQES$ defines a map
\begin{equation}
\mathcal E_{CC'\to AA'}=\frac{1}{|K|}\sum_k D_k\otimes E_k^T,
\end{equation}
where the transpose $E_k^T$ is the map whose Kraus operators are the transposes of the Kraus operators of $E_k$ (in the standard basis). Our goal is to prove that this map essentially acts like the $U\bar U$-twirl. We decompose the space $\hi_C^{\otimes 2}$ as
\begin{equation}
\hi_C^{\otimes 2}=\C\ket{\phi+}\oplus\supp\Pi^-
\end{equation}
which induces a decomposition of 
\begin{align}
\opr(\hi_C^{\otimes 2})
&=\C\proj{\phi^+}\oplus\left\{\ketbra{\phi^+}{v}\Big|\braket{\phi^+}{v}=0\right\}\nonumber\\
&\oplus\left\{\ketbra{v}{\phi^+}\Big|\braket{\phi^+}{v}=0\right\}\oplus\left\{X\in B\Big|\bra{\phi^+}X=X\ket{\phi^+}=0\right\}.
\end{align}
On the first and last direct summands, the correct behavior of $\mathcal E$ is easy to show: the first one corresponds to the identity, and the last one to the non-identity channels $\Lambda$ with $p_=(\Lambda)=0$. For the remaining two spaces, we employ \expref{Lemma}{lem:SKQES-char} which shows that the encryption map of any valid encryption scheme has the form of appending an ancillary mixed state and then applying an isometry. Evaluating $\mathcal E(\ketbra{\phi+}{v})$ for $\braket{\phi^+}{v}=0$ reduces to evaluating the adjoint of the average encryption map, $E^\dagger_K$, on traceless matrices. It is, however, easy to verify that $$\tr_A\mathcal E_{CC'\to AA'}(\sigma_C\otimes(\cdot)_{C'})=(E_K^T)_{C'\to A'}$$ for any $\sigma_C$. This can be used to prove $E_K=\langle\tau_C\rangle$ by observing that $\bra{\phi^+}_{CC'}\sigma_C\otimes\rho_{C'}\ket{\phi^+}_{CC'}=\tr(\sigma_C\rho_{C})$, so for rank-deficient $\rho$ we can calculate $\mathcal E_{CC'\to AA'}(\sigma_C\otimes(\cdot)_{C'})$ using what we have already proven.

The other direction is proven by a simple application of  \expref{Lemma}{lem:DP-CPTPtensCP}.
\qed
\hspace{.4cm}

The fact that $\ITNM$ is equivalent to 2-designs (for unitary schemes) is a straightforward consequence of the above.

\begin{proof} (of \expref{Theorem}{thm:USKQES-NM-2design})
First, assume $(\tau_K, E, D)$ is a unitary $\ITNM$ $\SKQES$ with $E_k=U_k(\cdot)U_k^\dagger$. Then it has $|C|=|A|$, and $D_K(\tau_C)=\tau_A$, so the conclusion of \expref{Theorem}{thm:effective-char} in this case (i.e., Equation \eqref{eq:effective-char}) is exactly the condition for $\{U_k\}$ to be an exact channel twirl design and therefore an exact 2-design. If $(\tau_K, E, D)$, on the other hand, is a unitary $\SKQES$ and $\{U_k\}$ is a 2-design, then Equation \eqref{eq:effective-char} holds and the scheme is therefore $\ITNM$ according to \expref{Theorem}{thm:effective-char}.
\end{proof}

\subsubsection{Relationship to ABW non-malleability.}\label{sec:ABW-exact}
Ambainis, Bouda and Winter give a different definition of non-malleability, expressed in terms of the effective maps that an adversary can apply to the plaintext by acting on the ciphertext produced from encrypting with a random key~\cite{ambainis2009nonmalleable}. According to their definition, a scheme is non-malleable if the adversary can only apply maps from a very restricted class \emph{when averaging over the key, and without giving side information to the active adversary}. Let us recall their definition here. 

First, given a \SKQES~$(\tau_K, E,D)$, we define the set $S := \{ D_K(\sigma_C) \,|\, \sigma_C \in \opr(\hi_C)\}$ consisting of all valid average decryptions. We then define the class $C^S_A$ of all ``replacement channels''. This is the set of CPTP maps belonging to the space
\begin{equation}
	\spa_{\R}\{\mathrm{id}_A, (X\mapsto \tr(X)\sigma_A) : \sigma_A\in S\}\,.
\end{equation}
We then make the following definition, which first appeared in~\cite{ambainis2009nonmalleable}.
\begin{defn}[ABW non-malleability]\label{def:ABWNM}
	A \SKQES~$(\tau_K, E,D)$ is ABW-non-malleable (\ABWNM) if it is \ITS, and for all channels $\Lambda_{C\to C}$, we have
	 \begin{equation}\label{eq:abw}
	\tr_K \left[D_{CK\to AK} \circ \Lambda_{C\to C} \circ E_{AK\to CK}(\,\cdot\,\otimes\tau_K)\right] \,\in\, C_A^S.
	\end{equation}
\end{defn}

As indicated in~\cite{ambainis2009nonmalleable}, an approximate version of Equation \eqref{eq:abw} is obtained by considering the diamond-norm distance between the effective channel and the set $C_A^S$; this implies the possibility of an auxiliary reference system, which is denoted $R$ in $\ITNM$. We emphasize that this reference system is not under the control of the adversary. In particular, \ABWNM~does not allow for adversaries which maintain \emph{and actively use} side information about the plaintext system. 

Another notable distinction is that~\cite{ambainis2009nonmalleable} includes a secrecy assumption in the definition of an encryption scheme; under this assumption, it is shown that a unitary \SKQES~is \ABWNM~if and only if the encryption unitaries form a 2-design. By our \expref{Theorem}{thm:USKQES-NM-2design}, we see that \ITNM~and \ABWNM~are equivalent in the case of unitary schemes. So, in that case, $\ABWNM$ actually ensures a much stronger security notion than originally considered by the authors of~\cite{ambainis2009nonmalleable}.

In the general case, $\ITNM$ is strictly stronger than $\ABWNM$. First, by comparing the conditions of \expref{Definition}{def:ABWNM} to Equation \eqref{eq:effective-char}, we immediately get the following corollary of \expref{Theorem}{thm:effective-char}.
\begin{cor}\label{cor:NM-implies-ABWNM}
	If a $\SKQES$ satisfies $\ITNM$, then it also satisfies \ABWNM.
\end{cor}
Second, we give a separation example which shows that \ABWNM~is highly insecure; in fact, it allows the adversary to ``inject'' a plaintext of their choice into the ciphertext. This is insecure even under the classical definition of information-theoretic non-malleability of \cite{kawachi2011characterization}. We now describe the scheme and this attack.
\begin{ex}\label{ex:injection}
Suppose $(\tau_K, E, D)$ is a \SKQES~that is both $\ITNM$ and $\ABWNM$. Define a modified scheme $(\tau_K, E', D')$, with enlarged ciphertext space $\hi_{C'} = \hi_{C}\oplus\hi_{\hat A}$ (where $\hi_{\hat A}\cong\hi_A$) and encryption and decryption defined by
\begin{align*}
E'(X) &= E(X)_{C}\oplus 0_{\hat A}\\
D'(X) &= D_{CK\to AK}(\Pi_{C}X\Pi_{C})+ \mathrm{id}_{\hat AK\to AK}(\Pi_{\hat A}X\Pi_{\hat A})\,.
\end{align*}
Then $(\tau_K, E', D')$ is \ABWNM~but not \ITNM.
\end{ex}
While encryption ignores $\hi_{\hat A}$, decryption measures if we are in $C$ or $\hat A$ and then decrypts (in the first case) or just outputs the contents (in the second case.) This is a dramatic violation of $\ITNM$: set $\hi_{\tilde B}\cong\hi_{A}$, trivial $B$ and $R$, and 
	\begin{equation}
	\Lambda_{C'\to C' \tilde B}(X)=\tr(X)0_{C}\oplus \proj{\phi^+}_{\hat A\tilde B}\,;
	\end{equation}
it follows that, for all $\rho$,
	\begin{equation}
		I(AR:\tilde B)_{\tilde\Lambda(\rho)}=2\log|A|\gg h(|C'|^{-2}) = h(p_=(\Lambda, \rho))\,.
	\end{equation}

Now let us show that $(\tau, E', D')$ is still $\ABWNM$. Let $\Lambda_{C'\to C'}$ be an attack, i.e., an arbitrary CPTP map. Then the effective plaintext map is
	\begin{equation}
		\tilde{\Lambda}_{A\to A}=D\circ \Lambda^C_{C\to C}\circ E+\Lambda^{\hat A}_{C\to A}\circ E,
	\end{equation}
	where $\Lambda^C(X_C)=\Pi_C\Lambda(X_C\oplus 0_{\hat A})\Pi_C$ and $\Lambda^{\hat A}(X_C)=\mathrm{id}_{\hat{A}\to A}(\Pi_{\hat A}\Lambda(X_C\oplus 0_{\hat A})\Pi_{\hat A})$. Since $(\tau, E, D)$ is $\ITS$ (\expref{Theorem}{thm:ITNMtoITS}), there exists a fixed state $\rho^0_C$ such that $E_K(\rho_A)=\rho^0_C$ for all $\rho_A$. Since $(\tau, E, D)$ is $\ABWNM$, we also know that
	$$
	\tr_K\circ D\circ \Lambda^C_{C\to C}\circ E=\tilde{\Lambda}_1 \in C_A^S\,,
	$$
with $S=\{ D_K(\sigma_C)\,|\,\sigma_C \in \opr(\hi_C)\}$. We therefore get
	\begin{equation}
		\tilde{\Lambda}_{A\to A}=\tilde{\Lambda}_1+\langle \Lambda^{\hat A}(\rho^0_C)\rangle\in C_A^{S'},
	\end{equation}
with $S'=\{ D'_K(\sigma_{C'})\,|\,\sigma_{C'} \in \opr(\hi_{C'})\}.$ This is true because $S'$ contains all constant maps, as $D'_K(0_{C}\oplus\rho_{\hat A})=\rho_A$.

\section{The approximate setting}\label{sec:appr}

We now consider the case of approximate non-malleability. Approximate schemes are relevant for several reasons. First, an approximate scheme with negligible error can be more efficient than an exact one: the most efficient construction of an exact 2-design requires a quantum circuit of $O(n\log n\log\log n)$ gates \cite{cleve2016near}, where approximate 2-designs can be achieved with linear-length circuits \cite{dankert2009exact}. Second, in practice, absolutely perfect implementation of all quantum gates is too much to expect---even with error-correction. Third, when passing to authentication one must allow for errors, as it is always possible for the adversary to escape detection (with low probability) by guessing the secret key. 

For all these reasons, it is important to understand what happens when the perfect secrecy and perfect non-malleability requirements are slightly relaxed. In this section, we show that our definitions and results are stable under such relaxations, and prove several additional results for quantum authentication.
We begin with the approximate-case analogue of perfect secrecy. 

\begin{defn}[Approximate secrecy]\label{def:eps-ITS}
	Fix $\varepsilon > 0$. A \SKQES~$(\tau_K, E, D)$ is $\varepsilon$-approximately secret ($\epsilon$-\ITS) if, for any $\hi_B$ and any $\rho_{AB}$, setting $\sigma_{CBK}=E(\rho_{AB}\otimes \tau_K)$ implies $I(C:B)_\sigma \leq \varepsilon.$
\end{defn}

Analogously to the exact case, unitary schemes satisfying approximate secrecy are equivalent to approximate one-designs (see \expref{Appendix}{appendix:secrecy}). 

\subsection{Approximate non-malleability}

\subsubsection{Definition.}

We now define a natural approximate-case analogue of \ITNM, i.e., \expref{Definition}{def:qNM}. Let us briefly recall the context. The malleability scenario is described by systems $A$, $C$, $B$ and $R$ (respectively, plaintext, ciphertext, side-information, and reference), an initial tripartite state $\rho_{ABR}$, and an attack channel $\Lambda_{CB\to C\tilde B}$. Given this data, we have the effective channel $\tilde \Lambda_{AB \to A\tilde B}$  defined in Equation \eqref{eq:effective-channel} and the ``unavoidable attack'' probability $p_=(\Lambda, \rho)$ defined in Equation \eqref{eq:p-equals}. The new definition now simply relaxes the requirement on the increase of the adversary's mutual information.

\begin{defn}[Approximate non-malleability]\label{def:eps-qNM}
A \SKQES~$(\tau_K,E, D)$ is $\varepsilon$-non-malleable ($\varepsilon$-\ITNM) if for any state $\rho_{ABR}$ and any CPTP map $\Lambda_{CB \to C\tilde B}$, we have 
	\begin{equation}\label{eq:eps-ITNM-condition}
	I(AR:\tilde B)_{\tilde\Lambda(\rho)}
	\leq I(AR:B)_\rho + h(p_{=}(\Lambda,\rho))+\varepsilon.
	\end{equation}
\end{defn}

We record the approximate version of \expref{Proposition}{thm:ITNMtoITS}, i.e., non-malleability implies secrecy. The proof is a straightforward adaptation of the exact case.
\begin{prop}\label{thm:eps-ITNMtoITS}
	Let $(\tau_K,E, D)$ be an $\varepsilon$-\ITNM~\SKQES. Then $(\tau_K,E, D)$ is $2\varepsilon$-\ITS.
\end{prop}

\subsubsection{Non-malleability with approximate designs.}

Continuing as before, we now generalize the characterization theorems of non-malleability (\expref{Theorem}{thm:effective-char} and \expref{Theorem}{thm:USKQES-NM-2design}) to the approximate case.

\begin{thm}\label{thm:eps-effective-char}
	Let $(\tau, E, D)$ be a \SKQES~with ciphertext dimension $|C|=2^{m}$ and $r>0$ a sufficiently large constant. Then the following holds:
	\begin{enumerate}
	\item If $(\tau, E, D)$ is $2^{-r m}$-$\ITNM$, then for any attack $\Lambda_{CB\to C\tilde B}$, the effective map $\tilde \Lambda_{AB\to A\tilde B}$ is $2^{-\Omega(m)}$-close (in diamond norm) to
	\begin{equation*}\label{eq:eps-effective-char}
	\tilde \Lambda^{\mathrm{exact}}_{AB\to A\tilde B}=\id_A\otimes \Lambda'_{B\to\tilde B}+\frac{1}{|C|^2-1}\left(|C|^2\left\langle D_K(\tau)\right\rangle-\id\right)_A\otimes \Lambda''_{B\to\tilde B},
	\end{equation*}
	with $\Lambda'$, $\Lambda''$ as in \expref{Theorem}{thm:effective-char}.
	
	\item Suppose that $\log|R| = O(2^m)$, where $R$ is the reference register in \expref{Definition}{def:eps-qNM}. Then there exists a constant $r$, such that if every attack $\Lambda_{CB\to C\tilde B}$ results in an effective map that is $2^{-r m}$-close to $\tilde \Lambda^{\mathrm{exact}}$, then the scheme is $2^{-\Omega(m)}$-\ITNM.
\end{enumerate}
\end{thm}
This theorem is proven with explicit constants in \expref{Appendix}{app:proofs} as \expref{Theorem}{thm:eps-effective-char-app}. The condition on $R$ required for the second implication is necessary, as the relevant mutual information can at worst grow proportional to the logarithm of the dimension according to the Alicki-Fannes inequality (\expref{Lemma}{lem:fannes}). This is not a very strong requirement, as it should be relatively easy for the honest parties to put a bound on their total memory.

Next, we record the corollary which states that, for unitary schemes, approximate non-malleability is equivalent to encryption with an approximate 2-design. The proof proceeds as in the exact case, now starting from \expref{Theorem}{thm:eps-effective-char}.
\begin{thm}\label{thm:eps-USKQES-NM-2design}
Let $\Pi = (\tau_K, E, D)$ be a unitary $\SKQES$ for $n$-qubit messages and $f:\N\to\N$ a function that grows at most exponential. Then there exists a constant $r>0$ such that
\begin{enumerate}
\item If $\{E_k\}$ is a $\Omega(2^{-rn})$-approximate 2-design and  $\log|R|\le f(n)$, then $\Pi$ is $2^{-\Omega(n)}$-$\ITNM$.
\item If $\Pi$ is $\Omega(2^{-rn})$-$\ITNM$, then $\{E_k\}_{k\in K}$ is a $2^{-\Omega(n)}$-approximate 2-design.
\end{enumerate}
\end{thm}

\subsubsection{Relationship to approximate ABW.}

Recall that, in \expref{Section}{sec:ABW-exact}, we discussed the relationship between our notion of exact non-malleability and that of Ambainis et al.~\cite{ambainis2009nonmalleable} (i.e., \ABWNM.) As we now briefly outline, our conclusions carry over to the approximate case without any significant changes. 

As described in Equation (3'') of~\cite{ambainis2009nonmalleable}, one first relaxes the notion of \ABWNM~appropriately by requiring that the containment \eqref{eq:abw} in \expref{Definition}{def:ABWNM} holds up to $\varepsilon$ error in the diamond-norm distance. In the unitary case, both definitions are equivalent to approximate 2-designs (by the results of~\cite{ambainis2009nonmalleable}, and our \expref{Theorem}{thm:eps-USKQES-NM-2design}). In the case of general schemes, the plaintext injection attack described in \expref{Example}{ex:injection} again shows that approximate $\ABWNM$ is insufficient, and that approximate $\ITNM$ is strictly stronger.

\subsection{Authentication}\label{sec:authentication}

We now consider the well-studied task of information-theoretic quantum authentication, and explain its connections to non-malleability.

\subsubsection{Definitions.}

Our definitions of authentication will be faithful to the original versions in~\cite{dupuis2012actively, Garg2017}, with one slight modification. When decryption rejects, our encryption schemes (\expref{Definition}{def:SKQES}) output $\bot$ in the plaintext space, rather than setting an auxiliary qubit to a ``reject'' state. These definitions are equivalent in the sense that one can always set an extra qubit to ``reject'' conditioned on the plaintext being $\bot$ (or vice-versa). Nonetheless, as we will see below, this mild change has some interesting consequences.

We begin with the definition of Dupuis, Nielsen and Salvail~\cite{dupuis2012actively}, which demands that the effective average channel of the attacker ignores the plaintext.

\begin{defn}[DNS Authentication~\cite{dupuis2012actively}]\label{def:DNS-auth}
A \SKQES~ $(\tau_K, E, D)$ is called $\varepsilon$-DNS-authenticating if, for any CPTP-map $\Lambda_{CB\to CB'}$, there exists CP-maps 
$\Lambda^\acc_{B\to \tilde B}$ and $\Lambda^\rej_{B\to \tilde B}$ such that $\Lambda^\acc + \Lambda^\rej$ is\,\footnote{Note that there is a typographic error in \cite{dupuis2012actively} and \cite{Broadbent2016} at this point of the definition. In those papers, the two effective maps are asked to sum to the identity, which is impossible for many obvious choices of $\Lambda$.} TP, and for all $\rho_{AB}$ we have
	\begin{equation}\label{eq:DNS-auth}
	\bigl\| \tr_K D(\Lambda(E(\rho_{AB}\otimes \tau_K))) - (\Lambda^\acc(\rho_{AB}) + \proj{\bot}\otimes \Lambda^\rej(\rho_{B}))\bigr\|_1\le \varepsilon\,.
	\end{equation}
\end{defn}

An alternative definition was recently given by Garg, Yuen and Zhandry~\cite{Garg2017}. It asks that, \emph{conditioned on acceptance}, with high probability the effective channel is close to a channel which ignores the plaintext.

\begin{defn}[GYZ Authentication~\cite{Garg2017}]\label{def:auth}
	A \SKQES~ $(\tau_K, E, D)$ is called $\varepsilon$-GYZ-authenticating if, for any CPTP-map $\Lambda_{CB\to CB'}$, there exists a CP-map $\Lambda^\acc_{B\to \tilde B}$ such that for all $\rho_{AB}$
	\begin{equation}\label{eq:auth}
	\bigl\|\Pi_\acc\,D(\Lambda(E(\rho_{AB}\otimes \tau_K)))\,\Pi_\acc - \Lambda^\acc(\rho_{AB})\otimes \tau_K\bigr\|_1\le \varepsilon\,.
	\end{equation}
	Here $\Pi_\acc$ is the acceptance projector, i.e. projection onto $\hi_A$ in $\hi_A\oplus\C\ket\bot$.
\end{defn}
A peculiar aspect of the original definition in~\cite{Garg2017} is that it does not specify the outcome in case of rejection, and is thus stated in terms of trace non-increasing maps. Of course, all realistic quantum maps must be CPTP; this means that the designer of the encryption scheme must still declare what to do with the contents of the plaintext register after decryption. Our notion of decryption makes one such choice (i.e., output $\bot$) which seems natural.

\subsubsection{GYZ authentication implies DNS authentication.}
A priori, the relationship between Definition 2.2 in~\cite{dupuis2012actively} and Definition 8 in~\cite{Garg2017} is not completely clear. On one hand, the latter is stronger in the sense that it requires success with high probability (rather than simply on average.) On the other hand, the former makes the additional demand that the ciphertext is untouched even if we reject. As we will now show, with our slight modification, we can prove that GYZ-authentication implies DNS-authentication.
\begin{thm}
	Let  $(\tau, E,D)$ be $\varepsilon$-totally authenticating for sufficiently small $\varepsilon$. Then it is $O(\sqrt{\varepsilon})$-DNS authenticating.
\end{thm}
\begin{proof}
	Let $\Lambda_{CB\to C\tilde B}$ be a CPTP map and $\varepsilon\le 62^{-2}$. By Definition \ref{def:auth} there exists a CP map $\Lambda'_{B\to \tilde B}$ such that for all states $\rho_{AB}$,
	\begin{equation}\label{eq:authdef}
	\left\|\Pi_a D(\Lambda(E(\rho_{AB}\otimes \tau_K)))\Pi_a-\Lambda'(\rho_{AB}\otimes \tau_K)))\right\|_1\le \varepsilon\,.
	\end{equation}
	Assume for simplicity that $D=M_\bot\circ D$, where $M_\bot$ measures the rejection symbol versus the rest. (otherwise we can define a new decryption map that way.) Define the CP maps
	\begin{align*}
	\Lambda^{(1)}_{AB\to \tilde B} 
	&=\tr_A\Pi_a\tilde{\Lambda}(\cdot)\\
	\Lambda^{(2)}_{AB\to \tilde B}
	&=\bra{\bot}_A\tilde\Lambda(\cdot)\ket{\bot}_A\\
	\Lambda''_{B\to \tilde B}
	&=\tr_C\Lambda(E_K(\tau_A)\otimes(\cdot)).
	\end{align*}
By Theorem 15 in \cite{Garg2017} we have 
	\begin{equation}
	\left|E_K(\rho_{ABR})-E_K(\tau_A)\otimes\rho_{BR}\right\|_1\le 14\sqrt{\varepsilon},
	\end{equation}
	which implies that
	\begin{equation}\label{eq:decomp-1}
	\left\|\tr_A\otimes\Lambda''-\tr_C\circ\Lambda\circ E_K\right\|_\diamond\le\hat\varepsilon := 14\sqrt{\varepsilon}.
	\end{equation}
	Note that
	\begin{align}\label{eq:decomp-2-1}
	\tr_C\!\circ\!\Lambda\!\circ\! E_K 
	&=\tr_{CK}\!\circ\!\Lambda\!\circ\! E((\cdot)\otimes\tau_K) \nonumber\\
	&=\tr_{AK}\!\circ\! D\!\circ\!\Lambda\!\circ\! E((\cdot)\otimes\tau_K)=\tr_A\!\circ\!\tilde{\Lambda}.
	\end{align}
	On the other hand, we also have that, by Equation \eqref{eq:authdef},
	\begin{align}\label{eq:decomp-2-2}
	\bigl\|\tr_A\circ\tilde{\Lambda}-\tr_A\otimes\Lambda'-\Lambda^{(2)}\bigr\|\le
	\bigl\|\tr_A\left(\Pi_a\tilde{\Lambda}(\cdot)\right)-\Lambda'\bigr\|_\diamond\le\varepsilon
	\end{align}
Combining Equations \eqref{eq:decomp-1}, \eqref{eq:decomp-2-1} and \eqref{eq:decomp-2-2}, we get
	\begin{equation}\label{eq:diamondbound}
	\bigl\|\Lambda^{(2)}-\tr_A\otimes(\Lambda''-\Lambda')\bigr\|_\diamond\le \varepsilon+\hat\varepsilon.
	\end{equation}
	Now observe that
	\begin{equation}\label{eq:doesntdependonA}
	\left[\tr_A\otimes(\Lambda'-\Lambda'')_{B\to\tilde B}\right]\circ\Xi_{A\to A}=\tr_A\otimes(\Lambda'-\Lambda'')_{B\to\tilde B}
	\end{equation}
	For all CPTP maps $\Xi_{A\to A}$. We define $\Lambda'''_{B\to \tilde B}=\Lambda^{(2)}(\tau_A\otimes(\cdot))$ and calculate
	\begin{align*}
	\bigl\|\Lambda^{(2)}-\tr_A\otimes\Lambda'''\bigr\|_\diamond
	&\le \bigl\|\Lambda^{(2)}-\tr_A\otimes(\Lambda''-\Lambda')\bigr\|_\diamond\\
	&~~~+\bigl\|\tr_A\otimes(\Lambda''-\Lambda')-\tr_A\otimes\Lambda'''\bigr\|_\diamond\,,
	\end{align*}
by the triangle inequality for the diamond norm. Continuing with the calculation, 
	\begin{align}\label{eq:central-bound}
	\bigl\|\Lambda^{(2)}-\tr_A\otimes\Lambda'''\bigr\|_\diamond
	&\le \varepsilon+\hat{\varepsilon}+\bigl\|\tr_A\otimes(\Lambda''-\Lambda')-\tr_A\otimes\Lambda'''\bigr\|_\diamond\nonumber\\
	&= \varepsilon+\hat{\varepsilon}+\bigl\|\tr_A\otimes (\Lambda''-\Lambda')-\Lambda^{(2)}\circ \langle\tau_A\rangle_{A\to A}\bigr\|_\diamond\nonumber\\
	&= \varepsilon+\hat{\varepsilon}+\bigl\|\bigl[\tr_A\otimes (\Lambda''-\Lambda')-\Lambda^{(2)}\bigr]\circ \langle\tau_A\rangle_{A\to A}\bigr\|_\diamond\nonumber\\
	&\le 2(\varepsilon+\hat\varepsilon)=28\sqrt{\varepsilon}+2\varepsilon.
	\end{align}
	The first inequality above is Equation \eqref{eq:diamondbound}. The first equality is just a rewriting of the definition of $\Lambda'''$, and the second equality is Equation \eqref{eq:doesntdependonA}. Finally, the last inequality is due to Equation \eqref{eq:diamondbound} and the fact that the diamond norm is submultiplicative.
	
	We have almost proven security according to \expref{Definition}{def:DNS-auth}, as we have shown $\tilde\Lambda$ to be close in diamond norm to $\id_A\otimes\Lambda'+\big\langle\proj{\bot}\big\rangle\otimes\Lambda'''$. However, $\Lambda'+\Lambda'''$ is only approximately TP; more precisely, we have that for all $\rho_{ABR}$,
	\begin{align}\label{eq:tracebound}
	|\tr(\Lambda'+\Lambda''')(\rho_{ABR})-1|
	&=|\tr(\Lambda'+\Lambda'''-\Lambda)(\rho_{ABR})|\nonumber\\
	&\le|\tr(\Lambda'-\Lambda^{(1)})(\rho_{ABR})|+|\tr(\Lambda'''-\Lambda^{(2)})(\rho_{ABR})|\nonumber\\
	&\le 28\sqrt{\varepsilon}+3\varepsilon.
	\end{align}
	We therefore have to modify $\Lambda' + \Lambda''$ so that it becomes TP, while keeping the structure required for DNS authentication.
	Let $M_B=(\Lambda'+\Lambda''')^\dagger(\mathds 1_{\tilde B})$, and $\lambda_{\min}$ and $\lambda_{\max}$ its minimal and maximal eigenvalue. Then Equation \eqref{eq:tracebound} is equivalent to $\lambda_{\min}\ge 1-\eta$ and $\Lambda_{\max}\le 1+\eta$, where we have set $\eta := 28\sqrt{\varepsilon}+3\varepsilon$. Now define the corresponding CP-map, i.e., $\mathcal{M}(X)=M^{-1/2}XM^{-1/2}$. Note that $M$ is invertible for $\eta<1$ which follows from $\varepsilon\le 62^{-2}$. We bound
	\begin{align}\label{eq:Mbound}
	\left\|\mathcal{M}-\id\right\|_\diamond&=\sup_{\rho_{BE}}\bigl\|M^{-1/2}_B\rho_{BE} M^{-1/2}_B-\rho_{BE}\bigr\|_1\nonumber\\
	&\le \sup_{\rho_{BE}}\bigl\{\bigl\|M^{-1/2}_B\rho_{BE}\bigl(M^{-1/2}_B-\mathds 1_B\bigr)\bigr\|_1+\bigl\|\bigl(M^{-1/2}_B-\mathds 1_B\bigr)\rho_{AB}\bigr\|_1\bigr\}\nonumber\\
	&\le \bigl(\bigl\|M^{-1/2}\bigr\|_\infty+1\bigr)\bigl\|M^{-1/2}-\mathds 1\bigr\|_\infty\nonumber\\
	&=(1+\lambda_{\min}^{-1/2})\max(1-\lambda_{\max}^{-1/2}, \lambda_{\min}^{-1/2}-1)\nonumber\\
	&\le (1+(1-\eta)^{-1/2})\max\bigl[1-(1+\eta)^{-1/2},(1-\eta)^{-1/2}-1\bigr]\nonumber\\
	&= (1+(1-\eta)^{-1/2})((1-\eta)^{-1/2}-1)
	=\frac{\eta}{1-\eta}\le 2\eta
	\end{align}
	The first inequality is the triangle inequality of the trace norm. The second inequality follows by three applications of H\" older's inequality with $p=1$ and $q=\infty$. The last inequality follows from the assumption $\varepsilon\le 62^{-2}$. The second to last equality holds because $\sqrt{1+x}\le 1+\frac{x}{2}$ and $(1+x/2)^{-1}\ge 1-x/2$ for $x\in[-1,1]$ imply
	\begin{align}
	(1-\eta)^{-1/2}-1\ge&(1-\eta/2)^{-1}-1\ge\eta/2
	\end{align}
	and 
	\begin{align}
	1-(1+\eta)^{-1/2}\le&1-(1+\eta/2)^{-1}\le\eta/2.
	\end{align}
	Altogether we have
	\begin{align}
	&\bigl\|\bigl(\id_A\otimes\Lambda'+\big\langle\proj{\bot}\big\rangle\otimes\Lambda'''\bigr)\circ\mathcal M-\tilde{\Lambda}\bigr\|\nonumber\\
	&\le\bigl\|\bigl(\id_A\otimes\Lambda'+\big\langle\proj{\bot}\big\rangle\otimes\Lambda'''-\tilde\Lambda\bigr)\circ\mathcal M\bigr\|_\diamond
	+\bigl\|\tilde\Lambda\circ\bigl(\mathcal M-\id\bigr)\bigr\|_\diamond\nonumber\\
	&\le\bigl\|\id_A\otimes\Lambda'+\big\langle\proj{\bot}\big\rangle\otimes\Lambda'''-\tilde\Lambda\bigr\|_\diamond\|\mathcal M\|_\diamond
	+\bigl\|\mathcal M-\id\bigr\|_\diamond\nonumber\\
	&=\bigl(\bigl\|\id_A\otimes\Lambda'-\Pi_a\tilde\Lambda\Pi_a\bigr\|_\diamond
	+\bigl\|\big\langle\proj{\bot}\big\rangle\otimes\Lambda'''-\proj{\bot}\otimes\Lambda^{(2)}\bigr\|_\diamond\bigr)\|M^{-1}\|_\infty\nonumber\\
	&~~~+\left\|\mathcal M-\id\right\|_\diamond\nonumber\\
	&\le 2(\varepsilon+28\sqrt{\varepsilon}+2\varepsilon)+2\eta=4\eta
	\end{align}
	The first and second inequality are the triangle inequality and the submultiplicativity of the diamond norm. The third inequality is due to Equations \eqref{eq:authdef}, \eqref{eq:central-bound} and \eqref{eq:Mbound}, as well as $\varepsilon\le 62^{-2}$. For the first equality, note that it is easy to check that $\|\mathcal M\|_\diamond=\|M^{-1/2}\|_\infty^2=\lambda_{\min}^{-1}$.
	
\qed
\end{proof}

\subsubsection{Achieving GYZ authentication with two-designs.}

In \cite{Garg2017}, the authors provide a scheme for their notion of authentication based on unitary eight-designs. We now show that, in fact, an approximate 2-design suffices. This is interesting, as it implies that the well-known Clifford scheme (see e.g \cite{dupuis2010secure,Broadbent2016}) satisfies the strong security of \expref{Definition}{def:auth}. All of the previous results on authentication which use the Clifford scheme thus automatically carry over to this stronger setting. We remark that our proof is inspired by the reasoning based on Schur's lemma used in results on decoupling \cite{berta2011quantum,dupuis2014one,majenz2016catalytic,berta2016deconstruction}.

\begin{thm}\label{thm:2-design-auth}
	Let $\mathrm D=\left\{U_k\right\}_k$ be a $\delta$-approximate unitary 2-design on $\hi_C$. Let $\hi_C=\hi_{A}\otimes\hi_T$ and define
	\begin{align*}
	E_k(X_A) &= U_k\left(X_A\otimes \proj{0}_T\right)\left(U_k\right)^\dagger\\
	D_k(Y_C) &= \bra 0_T\left(U_k\right)^\dagger Y U_k\ket 0_T+\tr((\mathds 1_T-\proj 0_T)\left(U_k\right)^\dagger Y U_k)\proj{\bot}\,.
	\end{align*}
	Then the \SKQES~$(\tau_K, E, D)$ is $4(1/|T| + 3\delta)^{1/3}$-GYZ-authenticating.
\end{thm}

\begin{rem}
	The following proof uses the same simulator as the proof for the 8-design scheme in \cite{Garg2017}, called "oblivious adversary" there. The construction exhibited there is efficient given that the real adversary is efficient.	
	\end{rem}

\begin{proof}
	To improve readability, we will occasionally switch between adding subscripts to operators (indicating which spaces they act on) and omitting these subscripts.	
	
	We begin by remarking that it is sufficient to prove the GYZ condition (specifically, \expref{Equation}{eq:auth}) for pure input states and isometric adversary channels. Indeed, for a general state $\rho_{AB}$ and a general map $\Lambda_{CB\to C\tilde B}$, we may let $\rho_{ABR}$ and $V_{CB\to C\tilde BE}$ be the purification and Stinespring dilation, respectively. We then simply observe that the trace distance decreases under partial trace (see e.g. \cite{nielsen2010quantum}).
	
	Let $\rho_{AB}$ be a pure input state and 
	$$
	\Lambda_{CB\to C\tilde B}(X_{CB}) = V_{CB\to C\tilde B}X_{CB}V_{CB\to C\tilde B}^\dagger
	$$
	an isometry. We define the corresponding ``ideal'' channel $\Gamma_V$, and the corresponding ``real, accept'' channel $\Phi_k$, as follows:
	\begin{align}
	\left(\Gamma_V\right)_{B\to\tilde B}&=\frac{1}{|C|}\tr_CV\text{ and}\nonumber\\
	\left(\Phi_k\right)_{AB\to A\tilde B}&=\bra 0_T(U_k)^\dagger_C V_{CB\to C\tilde B} U_k\ket 0_T.
	\end{align}
	Note that for any matrix $M$ with $\|M\|_\infty\le 1$, the map $\Lambda_M(X)=M^\dagger XM$ is completely positive and trace non-increasing. We have
	\begin{equation}
	\left\|\Gamma_V\right\|_\infty\le \frac{1}{|C|}\sum_i\left\|\bra i V \ket i\right\|_\infty\le 1.
	\end{equation}

	We begin by bounding the expectation of $\left\|(\left(\Gamma_V\right)_{B\to\tilde B}-\left(\Phi_k\right)_{AB\to A\tilde B})\ket{\rho}_{AB}\right\|_2^2$, as follows. To simplify notation, we set $\sigma_{ABT} := \proj{\rho}_{AB}\otimes\proj 0_T$ to be the tagged state corresponding to plaintext (and side information) $\rho_{AB}$.
	\begin{align}\label{eq:authbound1}
	\frac{1}{|K|}&\sum_k\left\|(\Gamma_V-\Phi_k)\ket{\rho}\right\|_2^2
	=\frac{1}{|K|}\sum_k\bra{\rho}(\Gamma_V-\Phi_k)^\dagger(\Gamma_V-\Phi_k)\ket{\rho}\nonumber\\
	&=\frac{1}{|K|}\sum_k\tr\left[\sigma_{ABT} (U_k)^\dagger V^\dagger U_k\proj 0(U_k)^\dagger V U_k\right]\nonumber\\
	&~~~~~- 2\frac{1}{|K|}\sum_k\tr\left[\sigma_{ABT} (U_k)^\dagger V^\dagger U_k \Gamma_V\right]
	+ \bra{\rho}\left(\Gamma_V\right)^\dagger \Gamma_V\ket{\rho}\,.
	\end{align}
	First we bound the second term, using the fact that $\Gamma_V$ only acts on $B$.
	\begin{align}\label{eq:authbound4}
	\frac{1}{|K|}&\sum_k\tr\left[\sigma_{ABT} (U_k)^\dagger V^\dagger U_k \Gamma_V\right]
	= \frac{1}{|K|}\sum_k\tr\left[U_k\sigma_{ABT}(U_k)^\dagger V^\dagger  \Gamma_V\right]\nonumber\\
	&= \int\tr\left[\left(U \sigma_{ABT} U^\dagger+\Delta\right) V^\dagger  \Gamma_V\right]
	\ge \int\tr\left[U \sigma_{ABT} U^\dagger V^\dagger  \Gamma_V\right]-\delta\nonumber\\
	&= \int\tr\left[ \sigma_{ABT} U^\dagger V^\dagger  U\Gamma_V\right]-\delta
	= \bra{\rho}\left(\Gamma_V\right)^\dagger \Gamma_V\ket{\rho}-\delta\,.
	\end{align}
	In the above, the operator $\Delta$ is the ``error'' operator in the $\delta$-approximate 2-design. The second equality above follows from $\|\Delta\|_1 \leq \delta$ and the fact that a 2-design is also a 1-design; the inequality follows by H{\"o}lder's inequality, and the last step follows from Schur's lemma. 
	
	The first term of the RHS of Equation \eqref{eq:authbound1} can be simplified as follows. We will begin by applying the swap trick (\expref{Lemma}{lem:swap-trick}) $\tr [XY]=\tr [F  X\otimes Y]$ in the second line below. The swap trick is applied to register $CC'$, with the operators $X$ and $Y$ defined as indicated below.
	\begin{align}\label{eq:authbound2}
	&\frac{1}{|K|}\sum_k \tr\Bigl[\,\underbrace{\sigma_{ABT}(U_k)^\dagger_{C} V^\dagger_{C\tilde B\to CB} (U_k)_C\proj 0_T}_{X}\,\underbrace{(U_k)^\dagger_C V_{CB\to C\tilde B} (U_k)_C}_{Y}\,\Bigr]\nonumber\\
	&= \frac{1}{|K|}\sum_k \tr\left[\left(\sigma_{ABT}\otimes\proj 0_{T'}\right)\left(U_k^{\otimes 2}\right)_{CC'}V^\dagger_{C\tilde B\to CB}V_{C'B\to C'\tilde B}\left(U_k^{\otimes 2}\right)_{CC'}^\dagger F_{CC'}\right]\nonumber\\
	&= \frac{1}{|K|}\sum_k \tr\left[\left(U_k^{\otimes 2}\right)_{CC'}^\dagger\left(\sigma_{ABT}\otimes\proj 0_{T'}\right)\left(U_k^{\otimes 2}\right)_{CC'}V^\dagger_{C\tilde B\to CB}V_{C'B\to C'\tilde B} F_{CC'}\right]\nonumber\\
	&\le \int\tr\left[\left(U^{\otimes 2}\right)_{CC'}^\dagger\left(\sigma_{ABT}\otimes\proj 0_{T'}\right)U^{\otimes 2}_{CC'}V^\dagger_{C\tilde B\to CB}V_{C'B\to C'\tilde B} F_{CC'}\right]+\delta\nonumber\\
	&= \int \tr\left[\left(\sigma_{ABT}\otimes\proj 0_{T'}\right)U^{\otimes 2}_{CC'}V^\dagger_{C\tilde B\to CB}V_{C'B\to C'\tilde B}\left(U^{\otimes 2}\right)_{CC'}^\dagger F_{CC'}\right]+\delta.
	\end{align}
	The inequality above follows the same way as in \expref{Equation}{eq:authbound4}. Let $d=|C|$. We calculate the integral above using \expref{Lemma}{lem:Usquared}, as follows.
	\begin{equation}\label{eq:authbound3}
	\int U^{\otimes 2}V^\dagger_{C\tilde B\to CB}V_{C'B\to C'\tilde B}\left(U^{\otimes 2}\right)^\dagger\D U = \mathds 1_{CC'}\otimes R^{\mathds 1}_B+F_{CC'}\otimes R^F_B,
	\end{equation}
	where we have set
	\begin{align}
	R^{\mathds 1}_B=&\frac{1}{d(d^2-1)}\left(d^3\Gamma_V^\dagger \Gamma_V -d\mathds 1\right)\nonumber
	=\frac{1}{(d^2-1)}\left(d^2\Gamma_V^\dagger \Gamma_V -\mathds 1\right)\nonumber\\
	R^{F}_B=&\frac{1}{d(d^2-1)}\left(d^2\mathds 1-d^2\Gamma_V^\dagger \Gamma_V\right)\nonumber
	=\frac{d}{(d^2-1)}\left(\mathds 1-\Gamma_V^\dagger \Gamma_V\right).
	\end{align}
	plugging \eqref{eq:authbound3} into \eqref{eq:authbound2} and using \expref{Lemma}{lem:swap-trick} again, we get
	\begin{align}
	\int \tr&\left[\left(\sigma_{ABT}\otimes\proj 0_{T'}\right)U^{\otimes 2}_{CC'}V^\dagger_{C\tilde B\to CB}V_{C'B\to C'\tilde B}\left(U^{\otimes 2}\right)_{CC'}^\dagger F_{CC'}\right]\nonumber\\
	&= \tr\left[\left(\sigma_{ABT}\otimes\proj 0_{T'}\right)\left(\mathds 1_{CC'}\otimes R^{\mathds 1}_{B^2\to \tilde B^2}+F_{CC'}\otimes R^F_{B^2\to \tilde B^2}\right)F_{CC'}\right]\nonumber\\
	&= \tr\left[\proj{\rho}_{B}\left( R^{\mathds 1}_{B}+|A| R^F_{B}\right)\right]\nonumber\\
	&= \tr\left[\proj{\rho}_{B}\left(\frac{d(d-|A|)}{d^2-1}\left(\Gamma_V^\dagger \Gamma_V\right)_B+\frac{d|A|-1}{d^2-1}\mathds 1_B\right)\right]\,.
	\end{align}
	Now recall that $d=|A||T|$. Using the fact that $(a-1)/(b-1)\le a/b$ for $b \ge a$, we can give a bound as follows.
	\begin{align}
	\tr&\left[\proj{\rho}\left(\frac{d(d-|A|)}{d^2-1}\left(\Gamma_V^\dagger \Gamma_V\right)+\frac{d|A|-1}{d^2-1}\mathds 1\right)\right]\nonumber\\
	&= \frac{d|A|(|T|-1)}{d^2-1}\bra{\rho}\left(\Gamma_V^\dagger \Gamma_V\right)\ket{\rho}+\frac{d|A|-1}{d^2-1}\nonumber\\
	&\le \bra{\rho}\left(\Gamma_V^\dagger \Gamma_V\right)\ket{\rho}+\frac{1}{|T|}\,.
	\end{align}
	
	Putting everything together, we arrive at
	\begin{align}
	\frac{1}{|K|}\sum_k\left\|(\Gamma_V-\Phi_k)\ket{\rho}\right\|_2^2\le\frac{1}{|T|}+3\delta.
	\end{align}
	By Markov's inequality this implies
	\begin{equation}
	\mathbb{P}\left[\bigl\|(\Gamma_V-\Phi_k)\ket{\rho}\bigr\|_2^2>\alpha\left(\frac{1}{|T|}+3\delta\right)\right]\le\frac{1}{\alpha}
	\end{equation}
	which is equivalent to
	\begin{equation}
	\mathbb{P}\left[\bigl\|(\Gamma_V-\Phi_k)\ket{\rho}\bigr\|_2>\alpha^{1/2}\left(\frac{1}{|T|}+3\delta\right)^{1/2}\right]\le\frac{1}{\alpha},
	\end{equation}
	where the probability is taken over the uniform distribution on $\mathrm{D}$. Choosing $\alpha=(1/|T|+3\delta)^{-1/3}$ this yields
	\begin{equation}
	\mathbb{P}\left[\left\|(\Gamma_V-\Phi_k)\ket{\rho}\right\|_2>\left(\frac{1}{|T|}+3\delta\right)^{1/3}\right]\le \left(\frac{1}{|T|}+3\delta\right)^{1/3}.
	\end{equation}
	
	Let $S\subset D$ be such that $|S|/|\mathrm D|\ge 1-(1/|T|+3\delta)^{1/3}$ and 
	$\left\|(\Gamma_V-\Phi_k)\ket{\rho}\right\|_2\le(1/|T|+3\delta)^{1/3}$ for all $U_k\in S$. Using the easy-to-verify inequality $\|\proj{\psi}-\proj{\phi}\|_1\le 2\|\ket{\psi}-\ket{\phi}\|_2$, given as Lemma \ref{lem:1-norm-2-norm} in the supplementary material, we can bound
	\begin{align}
	\frac{1}{|K|}&\sum_{U_k\in\mathcal D}\left\|\Phi_k\proj{\rho}\left(\Phi_k\right)^\dagger-\Gamma_V\proj{\rho}\Gamma_V^\dagger\right\|_1\nonumber\\
	&\le\frac{1}{|S|}\sum_{U_k\in\mathcal S}\left\|\Phi_k\proj{\rho}\left(\Phi_k\right)^\dagger-\Gamma_V\proj{\rho}\Gamma_V^\dagger\right\|_1+2\left(\frac{1}{|T|}+3\delta\right)^{1/3}\nonumber\\
	&\le\frac{2}{|S|}\sum_{U_k\in\mathcal S}\left\|(\Gamma_V-\Phi_k)\ket{\rho}\right\|_2+2|T|^{-1/3}\nonumber\\
	&\le4\left(\frac{1}{|T|}+3\delta\right)^{1/3}.
	\end{align}
	This completes the proof for pure states and isometric adversary channels. As noted above, the general case follows.
	\qed
\end{proof}

As an example, one may set $|T|=2^{s}$  (i.e. $s$ tag qubits) and take an approximate unitary 2-design of accuracy $2^{-s}$. The resulting scheme would then be $\Omega(2^{-s/3})$-GYZ-authenticating.

A straightforward corollary of the above result is that, in the case of unitary schemes, adding tags to non-malleable schemes results in GYZ authentication. We leave open the question of whether this is the case for general (not necessarily unitary) schemes.

\begin{cor}\label{thm:cor-NM-GYZ}
	Let $(\tau, E, D)$ be a $2^{-rn}$-non-malleable unitary \SKQES~with plaintext space $A$. Define a new scheme $(\tau, E', D')$ with plaintext space $A'$ where $A = TA'$ and
	\begin{align*}
	E'(X)
	&= E(X\otimes\proj 0_T)\\
	D'(Y)
	&= \bra{0}_TD(Y)\ket 0_T+\tr\left[(\mathds 1_T-\proj 0_T)D(Y)\right]\proj \bot\,.
	\end{align*}
	Then there exists a constant $r>0$ such that $(\tau, E', D')$ is $2^{-\Omega(n)}$-GYZ-authenticating if $|T|=2^{\Omega(n)}$.
\end{cor}
The proof is a direct application of \expref{Theorem}{thm:eps-USKQES-NM-2design} (approximate non-malleability is equivalent to approximate 2-design) and \expref{Theorem}{thm:2-design-auth} (approximate 2-designs suffice for GYZ authentication.) We emphasize that, by \expref{Remark}{rem:efficient}, exponential accuracy requirements can be met with polynomial-size circuits.

\subsubsection{DNS authentication from non-malleability.}

We end with a theorem concerning the case of general (i.e., not necessarily unitary) schemes. We show that adding tags to a non-malleable scheme results in a \DNS-authenticating scheme. In this proof we will denote the output system of the decryption map by $\overline A$ to emphasize that it is $A$ enlarged by the reject symbol.

\begin{thm}
	Let $r$ be a sufficiently large constant, and let $(\tau, E, D)$ be an $2^{-rn}$-\ITNM~\SKQES~with $n$ qubit plaintext space $A$, and choose an integer $d$ dividing $|A|$. Then there exists a decomposition $A=TA'$ and a state $\ket{\psi}_T$ such that $|T| = d$ and the scheme $(\tau, E', D')$ defined by 
	\begin{align*}
	E^t(X)
	&= E(X\otimes\proj \psi_T)\\
	D^t(Y)
	&= \bra{\psi}_TD(Y)\ket \psi_T+\tr\left[(\mathds 1_T-\proj \psi_T)D(Y)\right]\proj \bot\,.
	\end{align*}
is $(4/|T|)+2^{-\Omega(n)}$-DNS-authenticating. 
\end{thm}
\begin{proof}
	We prove the statement for $\varepsilon=0$ for simplicity, the general case follows easily by employing \expref{Theorem}{thm:eps-effective-char} instead of \expref{Theorem}{thm:effective-char}.
	
	By \expref{Theorem}{thm:effective-char}, for any attack map $\Lambda_{CB\to C\tilde B}$, the effective map is equal to
	\begin{equation}
		\tilde{\Lambda}_{AB\to \overline A\tilde B}=\id_A\otimes\Lambda'_{B\to\tilde B}+\frac{1}{|C|^2-1}\left(|C|^2 \langle D_K(\tau_C)\rangle-\id\right)_{\overline A}\otimes\Lambda''_{B\to\tilde B}
	\end{equation}
	for CP maps $\Lambda'$ and $\Lambda''$ whose sum is TP. The effective map under the tagged scheme is therefore
	\begin{align*}
		\tilde{\Lambda}^t_{A'B\to \overline A'\tilde B}
		&= \bra{\psi}_T\tilde{\Lambda}_{AB\to \overline A\tilde B}((\cdot)\otimes \psi_T)\ket{\psi}_T\\
		&~~~+\tr\bigl[(\mathds 1_T-\psi_T)\tilde{\Lambda}_{AB\to \overline A\tilde B}((\cdot)\otimes \psi_T)\bigr]\proj \bot\\
		&=\left(\id_{A'}\right)_{A'\to\overline{A}'}\otimes\Lambda'_{B\to\tilde B}\\
		&~~~+\bigl(|C|^2 \big\langle \bigl(\bra{\psi}_TD_K(\tau_C)\ket{\psi}_T\bigr)_{A'}\oplus\beta\proj{\bot}\big\rangle-\id_{A'}\bigr)_{A\to \overline A'}\otimes \frac{\Lambda''_{B\to\tilde B}}{|C|^2-1}\\
	\end{align*}
	with $\beta= \tr\left[(\mathds 1-\psi)_TD_K(\tau_C)\right]$. We would like to say that, unless the output is the reject symbol, the effective map on $A$ is the identity. We do not know, however, what $D_K(\tau_C)$ looks like.
	Therefore we apply a standard reasoning that if a quantity is small \emph{in expectation}, then there exists at least one small instance. We calculate the expectation of $\tr\bra{\psi}_TD_K(\tau_C)\ket{\psi}_T$ when the decomposition $A=TA'$ is drawn at random according to the Haar measure,
	\begin{align}
		\int \tr\bra{\psi}U_A^\dagger D_K(\tau_C)U_A\ket{\psi}_T \D U_A
		&=\tr \left[\left(\int U_A\ket{\psi}_T\otimes\mathds 1_{A'}\psi U_A^\dagger \D U_A\right)D_K(\tau_C)\right] \nonumber\\
		&=\frac{\tr\mathds 1_A}{\tr\Pi_\acc}\tr \Pi_{\acc}D_K(\tau_C)\nonumber\\
		&\le1/|T|.
	\end{align}
	Hence there exists at least one decomposition $A=TA'$ and a state $\ket{\psi}_T$ such that $\hat{\gamma}:=\tr\bra{\psi}_TD_K(\tau_C)\ket{\psi}_T\le 1/|T|$. Define $\gamma=\max(\hat{\gamma}, |C|^{-2})$. For the resulting primed scheme, let 
	$$
	\Lambda_{\rej}:=\frac{(1-\gamma)|C|^2}{|C|^2-1}\Lambda''
	\qquad \text{and} \qquad
	\Lambda_{\acc}=\Lambda'+\frac{\gamma |C|^2-1}{|C|^2-1}\Lambda''\,.
	$$
	 We calculate the diamond norm difference between the real effective map an the ideal effective map,
	\begin{align}
		&\bigl\|\tilde{\Lambda}^t-\id\otimes \Lambda_\acc-\langle\proj{\bot}\rangle\otimes\Lambda_\rej\bigr\|_\diamond\nonumber\\
		&\le \bigl\|\id\otimes\Lambda'+\frac{1}{|C|^2-1}\bigl(|C|^2 \big\langle \bigl(\bra{\psi}D_K(\tau)\ket{\psi}\bigr)\big\rangle-\id\bigr)
		\otimes\Lambda''-\id\otimes \Lambda_\acc\bigr\|_\diamond\nonumber\\
		&~~~~~+\bigl\|\langle\proj{\bot\rangle}\otimes (1-\hat\gamma)|C|^2\Lambda'' / (|C|^2-1)-\langle\proj{\bot}\rangle\otimes\Lambda_\rej\bigr\|_\diamond\nonumber\\
		&\le(1+|C|^{-2})(|T|^{-1}+2|C|^{-2})\nonumber\\
		&= |T|^{-1}(1+(|A'||T|)^{-2})(1+2|A'|^{-2})\nonumber\\
		&\le 4|T|^{-1}
	\end{align}
	as desired\,.
\qed
\end{proof}
	
	\section{Open problems}
	
We close with a few natural open questions. 

First, one might consider other formulations of quantum non-malleability. We believe that our definition (i.e., $\ITNM$) captures the correct class of schemes; given this, there may still be more natural formulation of that class. For example, one may ask whether the definition must directly reference the attack map $\Lambda$ (our definition does do so, through the quantity $p_=(\Lambda, \rho)$.) Indeed, one may ask for a definition which refers only to newly created correlations between the adversary and the decryption; of course, one must still capture resistance to the attacks which led us to use $p_=(\Lambda, \rho)$ in our formulation. On the other hand, it may also be appealing to give a formulation which intensifies the dependence on the attack map, to reflect the fact that direct correlations between the adversary's initial side information and the plaintext are completely destroyed for $p_=(\Lambda, \rho)=0$.

Second, our characterization of non-malleable schemes (i.e., \expref{Theorem}{thm:eps-effective-char-app}) comes with an approximation penalty in the converse direction. It would be interesting to see whether the dimension factors in the penalty are necessary. A tightened entropic condition could simplify this task for one of the implications.

Finally, there are also open questions concerning the connection between quantum non-malleability and authentication. For instance, one may ask whether adding ``tags'' (i.e., $0$-qubits which are measured and checked during decryption) to an arbitrary non-malleable scheme results in a \GYZ-authenticating scheme.	
\section*{Acknowledgments}

The authors would like to thank Anne Broadbent, Alexander M\" uller-Hermes, Fr\'ed\'eric Dupuis and Christopher Broadbent2016n for helpful discussions. G.A. and C.M. acknowledge financial support from the European Research Council (ERC Grant Agreement 337603), the Danish Council for Independent Research (Sapere Aude) and VILLUM FONDEN via the QMATH Centre of Excellence (Grant 10059).

\appendix

\section{Some simple proofs}

In this section, we give some simple proofs about our definition of information-theoretic secrecy for quantum encryption. We also connect our notion of quantum non-malleability to classical non-malleability.

\subsection{Secrecy}\label{appendix:secrecy}

In the case of unitary schemes, information-theoretic quantum secrecy is equivalent to the unitary one-design property.

\begin{prop}\label{prop:one-design}
	A unitary \SKQES~$(\tau_K, E,D)$ is \ITS~if and only if $\mathrm D=\{U_k\}_{k\in \mathcal{K}}$ is a unitary 1-design, where $E_k(X)=U_kXU_k^\dagger$.
\end{prop}
\begin{proof}
	Let $\mathrm D$ be a unitary 1-design. Then the scheme is \ITS by Schur's Lemma (See, e.g., \cite{fulton1991representation}). Conversely, let $(\tau_K, E,D)$ be $\mathsf{ITS}$, i.e. $\mathcal T^{(1)}_{\mathcal D}(\rho_{AB})=(E_K)(\rho_{AB})=\mathcal T^{(1)}_{\mathcal D}(\rho_{A})\otimes\rho_B$. Suppose that there exist $\rho_A, \rho'_A$ such that $\mathcal T^{(1)}_{\mathcal D}(\rho_{A})\neq \mathcal T^{(1)}_{\mathcal D}(\rho'_{A})$. Then $I(A:B)_\sigma>0$ for $\sigma=E(\frac 1 2 \left(\rho_A\otimes\proj 0_B+\rho'_A\otimes \proj 1_B\right))$, which is a contradiction. This implies $T^{(1)}_{\mathcal D}(\rho_{A})=\sigma^0_A$ for all $\rho_A$. But $T^{(1)}_{\mathcal D}(\tau_{A})=\tau_A$, i.e. $\sigma^0_A=\tau_A$. The observation that the positive semidefinite matrices span the whole matrix space finishes the proof.
	\qed
\end{proof}

Next, we show that the two notions of perfect secrecy (\ITS~and \IND) are equivalent. Using Pinsker's inequality (\expref{Lemma}{lem:pinsker}) and the Alicki-Fannes inequality (\expref{Lemma}{lem:fannes}) one can also show that $\varepsilon$-\ITS~implies $(4\sqrt{2\varepsilon})$-\IND, and that $\varepsilon$-\IND~implies $(4h(\varepsilon)+6\varepsilon\log|A|)$-\ITS.
\begin{prop}
	A \SKQES~$(\tau_K,E,D)$ is \ITS~if and only if it is \IND.
\end{prop}
\begin{proof}
	Let $(\tau_K,E,D)$ be $\mathsf{ITS}$. Then there exists a $\sigma_0^C$ such that for all $\rho_{AB}$ $(E_K)_{KA\to KC}(\rho_{AB})=\sigma_C^0\otimes\rho_B$. This can be seen as follows: By definition every $\rho_{AB}$ is mapped to a product state. Suppose now $E_K(\rho_A)\neq E_K(\rho'_A)$, then $I(C:B)_{\sigma_{CB}}\neq 0$ with $\sigma_{CB}=E_K(\frac 1 2(\rho_A\otimes \proj 0_B+\rho'_A\otimes \proj 1_B))$, a contradiction. From this observation the $\mathsf{IND}$-property follows immediately.
	
	Conversely, let $(\tau_K,E,D)$ be $\mathsf{IND}$, that is in particular
	\begin{equation}
	\left\|(E_K)_{A\to C}(\rho^{(0)}_{A}-\rho^{(1)}_{A})\right\|_1=0
	\end{equation}
	for all $\rho^{(i)}_A$, $i=0,1$, i.e. $(E_K)_{A\to C}=\sigma^0_C\tr(\cdot)$ for some quantum state $\sigma^0_C$, as the set of quantum states spans all of $\opr(\hi_C)$. Now the $\mathsf{ITS}$-property follows immediately.
	\qed
\end{proof}

Next, we revisit in detail the example that shows that, in the approximate setting, there exist unitary schemes which can only be broken with access to side information. This is in contrast to the exact setting, where side information is unhelpful.

\begin{ex}
	Let $\mathrm{D}=\{\hat U^{(k)}\}_k$ be an exact unitary 2-design on a Hilbert space $\hi_A$ of even dimension $|A|=d$. Let $V_A$ be the unitary matrix with
	$$
	V_{j\,\,(d-j+1)}=i\,\cdot\,\mathrm{sign}(d-2j+1)
	$$
	for all $j$, and all other entries equal to zero. Set, for each $k$, 
	$$
	U^{(k)}=\hat U^{(k)}V\left(\hat U^{(k)}\right)^T\,.
	$$
	Define a \SKQES~by $E_k(X)=U^{(k)}X\left(U^{(k)}\right)^\dagger$ and $D_k=E_k^\dagger$.
\end{ex}

It is easy to check that $E_K(X_A)=\frac{1}{d-1}(d\tau_A-X^T_A)$ is the Werner-Holevo channel \cite{werner2002counterexample}. 
For any two quantum states $\rho_A, \rho'_A$,
\begin{align}
\left\|E_K(\rho_A)-E_K(\rho'_A)\right\|_1=&\frac{1}{d-1}\left\|{\rho'}^T_A-\rho^T_A\right\|_1\nonumber\\
\le&\frac{2}{d-1}.
\end{align}

On the other hand
\begin{align}
\left\|E_K-\langle\tau\rangle\right\|_\diamond=&\frac{1}{d-1}\left\|\langle\tau\rangle-\theta\right\|_\diamond\nonumber\\
\ge&\frac{1}{d-1}\left(\left\|\theta\right\|_\diamond-\left\|\langle\tau\rangle\right\|_\diamond\right)\nonumber\\
\ge&1.
\end{align}
The last step follows because the transposition map, here denoted by $\theta$, has diamond norm $d$\footnote{this is well known, the lower bound needed here can be obtained by applying the transposition to half of a maximally entangled state}.
By the definition of the diamond norm there exists a state $\rho_{AB}$ such that
\begin{equation}
\left\|E_K(\rho_{AB}-\tau_A\otimes\rho_B)\right\|_1=\left\|\left(E_K-\langle\tau_A\rangle\right)(\rho_{AB})\right\|_1\ge 1.
\end{equation}
In other words, we have exhibited a \SKQES~that is not $\varepsilon$-\IND~for any $\varepsilon<1/2$, but adversaries without side information achieve only negligible distinguishing advantage.

\subsection{Non-malleability}\label{appendix:nm}

We begin by recalling the following definition of classical, information-theoretic non-malleability. To this end we set down the notation in the classical case by giving a definition of a classical symmetric key encryption scheme.

\begin{defn}
	Let $\mathcal X, \mathcal C, \mathcal K$ be finite alphabets (finite sets). An symmetric key encryption scheme (SKES) $(K,E,D)$ is a key random variable (RV) $K$ on $\mathcal{K}$ together with a pair of stochastic maps $E: \mathcal X\times \mathcal K\to\mathcal C$ and $D: \mathcal C\times \mathcal K\to \mathcal X\cup\{\bot\}$ such that
	\begin{equation}
	E(\cdot,k)\circ D(\cdot, k)=\mathrm{id}_{\mathcal X}\footnote{This again is a slight abuse of notation, more correctly the composition of encryption and decryption yields the canonical injection of $X$ into $X\cup\{\bot\}$.},
	\end{equation}
	where $\mathrm{id}_{\mathcal X}$ denotes the identity function of $\mathcal X$ (\emph{correctness}). We write $E_k=E(\cdot, k)$ and analoguously $D_k$.
	
\end{defn}

\begin{defn}[Classical non-malleability \cite{kawachi2011characterization}]\label{def:CNM}
	A classical SKES scheme $(K,E, D)$ is non-malleable if the following holds. For all RVs $X$ on $\mathcal X$ independent of the key, and all RVs $\tilde C$ on $\mathcal C$ independent of the key given $X, C=E(X,K)$ such that $\mathbb P[\tilde C=C]=0$,
	\begin{equation}
	I(\tilde X:\tilde C|XC)=0,
	\end{equation}
	where $\tilde X=D(\tilde C, K)$.
\end{defn}

Next we show that, if a classical scheme satisfies our notion of information-theoretic quantum non-malleability, then it is also classically non-malleable according to \expref{Definition}{def:CNM} above.

\begin{prop}\label{prop:QNM-CNM}
	Let $(\tau_K,D,E)$ be a SKES~ embedded in to the quantum formalism in the standard way, which satisfies \ITNM. Then it is information theoretically non-malleable.
\end{prop}

\begin{proof}
	Let $B$ be a trivial system and $\rho_{AR}$ be a maximally correlated classical state. Let furthermore $\tilde B\cong CC'$, and let $\Lambda_{C\to C\tilde B}$ be a classical map from $C$ to $C$ that makes a copy of both input and output to the $\tilde B$ register and where $\tr \left(\proj i_C\otimes \mathds 1_{\tilde B}\right) \Lambda_{C\to C\tilde B}(\proj i_C)=0$ for all standard basis vectors $\ket i$ (this is the condition $\mathbb{P}[C=\tilde C]=0$). Then $p_=(\Lambda)=0$, and therefore according to the assumption
	\begin{equation}\label{eq:condmuinfonull}
	I(AR:\tilde B)_{\tilde\Lambda(\rho)}=0,
	\end{equation}
	as $B$ was trivial. Let $X,C,\tilde C, \tilde X$ be random variables corresponding to the systems $A$ in the beginning, $C$ after encryption, $C$ after the application of $\Lambda$ and $A$ after decryption, respectively. Then Equation \eqref{eq:condmuinfonull} reads
	\begin{align}
	0=&I(X\tilde X:C \tilde C)\nonumber\\=&I(X\tilde X:C)+I(X\tilde X:\tilde C|C)\\=&I(X\tilde X:C)+I(X:\tilde C|C)+I(\tilde X:\tilde C|CX),
	\end{align}
	as the input state was maximally classically correlated. The fact that all (conditional) mutual information terms above are non-negative finishes the proof.
	\qed
\end{proof}

\section{Technical lemmas}\label{appendix:technical}
In the following we prove some Lemmas in linear algebra and matrix analysis that we need in this article.
\begin{lem}\label{lem:genmirr}
	Let $X_{A\to B}\in L(\hi_A,\hi_B)$ be a linear operator from $A$ to $B$. Then
	\begin{equation}
	X_{A\to B}\ket{\phi^+}_{AA'}=\sqrt{\frac{|B|}{|A|}}X^T_{B'\to A'}\ket{\phi^+}_{BB'}.
	\end{equation}
\end{lem}
\begin{proof}
	\begin{align}
	X_{A\to B}\ket{\phi^+}_{AA'}=& \frac{1}{\sqrt{|A|}}\sum_{i=0}^{|A|-1}\sum_{j=0}^{|B|-1}X_{ji}\ket{j}_B\otimes\ket{i}_{A'}\nonumber\\
	=&\frac{1}{\sqrt{|A|}}\sum_{i=0}^{|A|-1}\sum_{j=0}^{|B|-1}X^T_{ij}\ket{j}_B\otimes\ket{i}_{A'}\nonumber\\
	=&\sqrt{\frac{|B|}{|A|}}X^T_{B'\to A'}\ket{\phi^+}_{BB'}.
	\end{align}
	\qed
\end{proof}

\begin{lem}\label{lem:1-norm-2-norm}
	Let $\ket{\psi},\ket{\phi}\in\hi$ be two vectors. Then
	\begin{equation}
	\|\proj{\psi}-\proj{\phi}\|_1\le 2\|\ket{\psi}-\ket{\phi}\|_2.
	\end{equation}
\end{lem}
\begin{proof}
	The trace norm distance of two pure states is given by \cite{nielsen2010quantum}
	\begin{equation}
	\|\proj{\psi}-\proj{\phi}\|_1=2\sqrt{1-|\braket{\phi}{\psi}|^2}.
	\end{equation}
	We bound
	\begin{align}
	\|\proj{\psi}-\proj{\phi}\|_1=&2\sqrt{1-|\braket{\phi}{\psi}|^2}\nonumber\\
	=&2\sqrt{(1-|\braket{\phi}{\psi}|)(1+|\braket{\phi}{\psi}|)}\nonumber\\
	\le&2\sqrt{2(1-|\braket{\phi}{\psi}|)}\nonumber\\
	\le&2\sqrt{2(1-\mathrm{Re}(\braket{\phi}{\psi}))}\nonumber\\
	=&2\|\ket{\psi}-\ket{\phi}\|_2
	\end{align}
	\qed
\end{proof}

\begin{lem}\label{lem:closeCJ2diamond}
	Let $\Lambda^{(i)}_{A\to B}$, $i=0,1$ be CPTP maps such that
	$$\left\|\eta_{\Lambda^{(0)}}-\eta_{\Lambda^{(1)}}\right\|_1\le \varepsilon.$$
	Then the two maps are also close in diamond norm,
	$$\left\|\Lambda^{(0)}_{A\to B}-\Lambda^{(1)}_{A\to B}\right\|_\diamond\le|A|\varepsilon.$$
\end{lem}
\begin{proof}
	The proof of this lemma is a simple application of the Hölder inequality. Let $\ket{\psi}_{AA'}=\sqrt{|A|}\psi_{A'}^{1/2}V_{A'}\ket{\phi^+}_{AA'}$ be an arbitrary pure state with $V_{A'}$ a unitary. Then we have
	\begin{align}
	\left\|(\Lambda^{(0)}_{A\to B}-\Lambda^{(1)}_{A\to B})(\proj{\psi})\right\|_1=&|A|\left\|\psi_{A'}^{1/2}V_{A'}(\eta_{\Lambda^{(0)}}-\eta_{\Lambda^{(1)}})V_{A'}^\dagger\psi_{A'}^{1/2}\right\|_1\nonumber\\
	\le&|A|\left\|\psi_{A'}^{1/2}\right\|_{\infty}^2\left\|V_{A'}\right\|_{\infty}^2\left\|\eta_{\Lambda^{(0)}}-\eta_{\Lambda^{(1)}}\right\|_1\le|A|\varepsilon.
	\end{align}
	\qed
\end{proof}

The next group of lemmas is concerned with entropic quantities.
\begin{lem}\label{lem:DP-CPTPtensCP}
	Let $\Lambda_{A\to A'}^{(i)}$ be CPTP maps and $\Lambda_{B\to B'}^{(i)}$, $i=1,...,k$ CP maps for $i=1,...,k$ such that $\sum_i\Lambda_{B\to B'}^{(i)}$ is trace preserving. Let $\Lambda^{(i)}_{AB\to A'B'}=\Lambda_{A\to A'}^{(i)}\otimes \Lambda_{B\to B'}^{(i)}$ and define the CPTP maps
	\begin{align}
	\Lambda_{AB\to A'B'C}=&\sum_{i=1}^k\Lambda_{AB\to A'B'}^{(i)}\otimes\proj i_C\text{ and}\nonumber\\
	\Lambda'_{B\to B'C}=&\sum_{i=1}^k\Lambda_{B\to B'}^{(i)}\otimes\proj i_C.
	\end{align}
	Then
	\begin{align}
	I(A':B')_{\Lambda(\rho)}\le I(A:B)_\rho+H(C|A)_{\Lambda'(\rho)}\le I(A:B)_\rho+H(C)_{\Lambda(\rho)}
	\end{align}
	for any quantum state $\rho_{AB}$.
\end{lem}
\begin{proof}
	Let $\rho_{AB}$ be a quantum state and define the following quantum states,  
	\begin{align}
	\sigma_{A'B'CC'}=&\sum_{i=1}^k\Lambda_{AB\to A'B'}^{(i)}(\rho_{AB})\otimes\proj i_C\otimes \proj i_{C'}\text{ and}\nonumber\\
	\sigma'_{AB'CC'}=&\sum_{i=1}^k\Lambda_{B\to B'}^{(i)}(\rho_{AB})\otimes\proj i_C\otimes \proj i_{C'},
	\end{align}
	i.e. $\sigma$ and $\sigma'$ are $\Lambda$ and $\Lambda'$ applied to $\rho$ with an extra copy of $C$.
	We bound
	\begin{align}
	I(A':B')_\sigma\le&I(A'C:B'C')_\sigma\nonumber\\
	\le&I(AC:B'C')_{\sigma'}\nonumber\\
	=&I(A:B'C')_{\sigma'}+I(C:B'C'|A)_{\sigma'}\nonumber\\
	=&I(A:B'C')_{\sigma'}+H(C|A)_{\sigma'}\nonumber\\
	\le&I(A:B)_\rho	+H(C|A)_{\sigma'}\nonumber\\
	\le&I(A:B)_\rho	+H(C)_{\sigma'}.
	\end{align}
	The first and second inequality are due to the data processing inequality of the quantum mutual information. The first equality is the chain rule for the quantum mutual information. The second equation is the fact that the mutual information of two copies of a classical system cannot be increased by adding systems on one side, relative to any conditioning system, and is equal to its entropy. The third inequality is due to the data processing inequality for the quantum mutual information again and the last inequality is the fact that conditioning can only decrease entropy. Using the definition of $\sigma$ and $\sigma'$ this implies the claim.\qed
\end{proof}

\begin{lem}[Fannes-Audenaert inequality, Alicki-Fannes inequality, \cite{fannes1973continuity,audenaert2007sharp,alicki2004continuity,wilde2013quantum}]\label{lem:fannes}
	\textcolor{white}{Fck LaTeX!}\\Let $\rho_{ABC}$ and $\rho'_{ABC}$ be tripartite quantum states such that
	\begin{equation}
	\|\rho_{ABC}-\rho'_{ABC}\|_1\le\varepsilon.
	\end{equation}
	Then the following continuity bounds hold for entropic quantities:
	\begin{align}
	|H(A)_\rho-H(A)_{\rho'}|\le &\frac \varepsilon 2 \log\left(|A|-1\right)+h\left(\frac \varepsilon 2\right)\le\varepsilon \log\left(|A|\right)+h(\varepsilon)\nonumber\\
	|H(A|B)_\rho-H(A|B)_{\rho'}|\le &4\varepsilon \log\left(|A|\right)+2h(\varepsilon)\nonumber\\
	|I(A:B)_\rho-I(A:B)_{\rho'}|\le &5\varepsilon\log\left(\min(|A|, |B|)\right)+3h(\varepsilon)\nonumber\\
	|I(A:B|C)_\rho-I(A:B|C)_{\rho'}|\le &8\varepsilon \log\left(\min(|A|, |B|)\right)+4h(\varepsilon).
	\end{align}
\end{lem}

\begin{lem}[Pinskers inequality, \cite{pinsker1960information}]\label{lem:pinsker}
	For quantum states $\rho_{AB}$ and $\sigma_{AB}$,
	\begin{align}
	D(\rho_A||\sigma_A)\ge&\frac 1 2\|\rho_A-\sigma_A\|_1^2\nonumber\\
	I(A:B)_\rho\ge &\frac 1 2\|\rho_{AB}-\rho_A\otimes \rho_B\|_1^2.
	\end{align}
\end{lem}

The next two lemmas concern the theory of unitary designs, and representation theory, respectively.

\begin{lem}\label{lem:channel-twirl-uubar-twirl} 
	Let $\dim\hi=d$, and let $\mathrm D\subset \mathrm{U}(\hi)$ be a finite set.
	\begin{itemize}
		\item If $\mathrm{D}$ is a $\delta$-approximate channel twirl, then it is also a $d\cdot\delta$-approximate $U$-$\overline{U}$ twirl design. 
		\item If $\mathrm{D}$ is a $\delta$-approximate $U$-$\overline{U}$ twirl design, then it is also a $d\cdot\delta$-approximate channel twirl design.
	\end{itemize}
\end{lem}

\begin{proof}
	Let $\mathrm{D}$ be a channel twirl design and $\rho_{AA'B}$ be a quantum state with $A'\cong A$ and $B$ arbitrary. Let furthermore $$\sigma_{AA'B}=(\mathds 1-\rho_{A'})^{1/2}\rho_{A'}^{-1/2}\rho_{AA'B}\rho_{A'}^{-1/2}(\mathds 1-\rho_{A'})^{1/2}.$$ Then $$\eta_{AA'BC}=\frac{1}{d}\left(\rho_{AA'B}\otimes\proj 0_C+\sigma_{AA'B}\otimes \proj 1_C\right)$$ is positive semidefinite, has trace 1 and $\eta_{A'}=\tau_{A'}$, i.e. it is the CJ-state of a quantum channel $\Lambda_{A\to ABC}$. By assumption we have that $$\left\|\mathcal T^{ch}_{\mathrm D}(\Lambda)-T^{ch}_\mathsf{Haar}(\Lambda)\right\|_\diamond\le\delta.$$ This implies in particular that the CJ states of the two channels have trace norm distance at most $\delta$, i.e.
	\begin{align}
	\delta\ge &\left\|\overline{\mathcal T}_{\mathrm D}(\eta)-\overline{\mathcal T}_\mathsf{Haar}(\eta)\right\|_1\nonumber\\
	=& \frac 1 d \left(\left\|\overline{\mathcal T}_{\mathrm D}(\rho)-\overline{\mathcal T}_\mathsf{Haar}(\rho)\right\|_1+\left\|\overline{\mathcal T}_{\mathrm D}(\sigma)-\overline{\mathcal T}_\mathsf{Haar}(\sigma)\right\|_1\right)\nonumber\\
	\ge& \frac 1 d \left\|\overline{\mathcal T}_{\mathrm D}(\rho)-\overline{\mathcal T}_\mathsf{Haar}(\rho)\right\|_1.
	\end{align}
	As $\rho$ was chosen arbitrarily this implies
	\begin{equation}
	\left\|\overline{\mathcal T}_{\mathrm D}-\overline{\mathcal T}_\mathsf{Haar}\right\|_{\diamond}\le d\cdot\delta.
	\end{equation}
	\qed
	Conversely, let $\mathrm D$ be a $U$-$\overline{U}$-twirl design, and let $\Lambda_{A\to A}$ be a CPTP map and $\ket{\psi}_{AA'}\in \hi_A\otimes\hi_{A'}$ be an arbitrary state vector with $\hi_{A'}\cong\hi_A$. For calculating the diamond norm, we can choose $\ket{\psi}_{AA'}=\sqrt{d}\psi_{A'}\ket{\phi^+}_{AA'}$. We bound
	\begin{align*}
	&\left\|\left[\left(\mathcal T^{ch}_{\mathrm D}-\mathcal T^{ch}_{\mathsf{Haar}}\right)(\Lambda)\right]_{A\to A}\left(\proj{\psi}_{AA'}\right)\right\|_1\nonumber\\=&d \left\|\left[\left(\mathcal T^{ch}_{\mathrm D}-\mathcal T^{ch}_{\mathsf{Haar}}\right)(\Lambda)\right]_{A\to A}\left(\psi_{A'}^{\frac{1}{2}}\proj{\phi^+}_{AA'}\psi_{A'}^{\frac{1}{2}}\right)\right\|_1\nonumber\\
	=&d\left\|\psi_{A'}^{\frac{1}{2}}\left[\left(\overline{\mathcal T}_{\mathrm D}-\overline{\mathcal T}_{\mathsf{Haar}}\right)\left(\eta_\Lambda\right)_{AA'}\right]\psi_{A'}^{\frac{1}{2}}\right\|_1\nonumber\\
	\le&d \|\psi_{A'}\|_\infty \left\|\left(\overline{\mathcal T}_{\mathrm D}-\overline{\mathcal T}_{\mathsf{Haar}}\right)\left(\eta_\Lambda\right)_{AA'}\right\|_1\nonumber\\
	\le& d\cdot\delta.
	\end{align*}
	Here we used the mirror lemma in the second equality, the Hölder inequality twice in the first inequality, and the fact that $\|\rho\|_\infty\le 1$ for any quantum state $\rho$, and the assumption, in the last inequality. 
\end{proof}

\begin{lem}\label{lem:Usquared}
	Let $M_{A^2B}$ be a matrix on $\hi_A^{\otimes 2}\otimes \hi_B$. Then we have the following formula for integration with respect to the Haar measure:
	\begin{equation}\label{eq:schur}
	\int U_A^{\otimes 2}M_{A^2B}(U_A^{\otimes 2})^\dagger \D U=\mathds 1_{A^2}\otimes R^{\mathds 1}_{B}+F_A\otimes R^F_{B},
	\end{equation}
	with 
	\begin{align}\label{eq:formulas}
	R^{\mathds 1}_{B}=&\frac{1}{d(d^2-1)}\left(d\tr_{A^2} M-\tr_{A^2} FM\right),\nonumber\\
	R^{F}_{B}=&\frac{1}{d(d^2-1)}\left(d\tr_{A^2} FM-\tr_{A^2} M\right).
	\end{align}
	
\end{lem}
\begin{proof}
	By Schur's lemma,
	\begin{equation}
	\int U_A^{\otimes 2}M_{A^2B}(U_A^{\otimes 2})^\dagger \D U=\Pi_{\bigwedge}\otimes R^{\bigwedge}+\Pi_{\bigvee}\otimes R^{\bigvee}
	\end{equation}
	for some matrices $R^{\bigwedge}$ and $R^{\bigvee}$, where $\Pi_{\bigwedge}$ is the projector onto the antisymmetric subspace and $\Pi_{\bigvee}$ the projector onto the symmetric subspace. As $\Pi_{\bigwedge}=(\mathds 1-F)/2$ and $\Pi_{\bigvee}=(\mathds 1+F)/2$, this implies the correctness of Equation \eqref{eq:schur}. The formulas \eqref{eq:formulas} follow by applying $\tr_A$ and $\tr_A(F(\cdot))$ to both sides of Equation \eqref{eq:schur} and solving the resulting system of 2 equations for $R^{\mathds 1}_{B}$ and $R^{F}_{B}$.
	\qed
\end{proof}

The final lemma characterizes CPTP maps that are invertible on their image such that the inverse is CPTP as well.

\begin{lem}\label{lem:SKQES-char}
	Let $(\tau_K,E, D)$ be a \SKQES. Then the encryption maps have the structure
	\begin{equation}
	\left(E_k\right)_{A\to C}=\left(V_k\right)_{A\hat C\to C}\left((\cdot)\otimes\sigma^{(k)}_{\hat C}\right)\left(V_k\right)_{A\hat C\to C}^\dagger,
	\end{equation}
	and the decryption maps hence must have the form
	\begin{align}\label{eq:Dkchar}
	&\left(D_k\right)_{C\to A}=\tr_{\hat C}\left[\Pi_{\supp \sigma^{k}}\left(V_k\right)_{A\hat C\to C}^\dagger\left(\cdot\right)\left(V_k\right)_{A\hat C\to C}\right]\nonumber\\
	&+\left(\hat D_k\right)_{A\hat C\to A}\left[(\mathds 1_{\hat C}-\Pi_{\supp \sigma^{k}})\left(V_k\right)_{A\hat C\to C}^\dagger\left(\cdot\right)\left(V_k\right)_{A\hat C\to C}(\mathds 1_{\hat C}-\Pi_{\supp \sigma^{k}})\right]
	\end{align}
	for some quantum states $\sigma^{(k)}_{\hat C}$, isometries $(V_k)_{C\to A\hat C}$, and some CPTP map $\hat D_k$.
\end{lem}

\begin{proof}
	Let $\ket{\phi}_{AA'}$ be some bipartite pure state. By the correctness of the scheme and the data processing inequality of the mutual information (see \expref{Lemma}{lem:DP-CPTPtensCP}),  we have that 
	\begin{align}
	2H(A')_\phi=&2H(A')_{E_k(\phi)} \ge I(A':C)_{E_k(\phi)}\nonumber\\
	\ge&I(A':A)_{D_k(E_k(\phi))}=2H(A')_{E_k(\phi)}=2H(A')_{\phi},
	\end{align}
	i.e. $2H(A')_\phi = I(A':C)_\psi$. The first inequality is an easy-to-check elementary fact. It is easy to see that this only holds if the purification of $\psi_{A'}=\phi_{A'}$ lies entirely in $C$, i.e. there exists an isometry $U^{(k)}_{C\to A\hat C}$ such that
	\begin{equation}
	U^{(k)}_{C\to A\hat C}E_k(\phi)\left(U^{(k)}_{C\to A\hat C}\right)^\dagger=\phi_{AA'}\otimes \sigma^{(k)}.
	\end{equation}
	note that by the linearity of $E_k$ and $U^{(k)}$, $\sigma^{(k)}$ cannot depend on $\phi$. As the state $\phi$ was arbitrary, this implies that
	\begin{equation}
	E_k=\left(U^{(k)}_{C\to A\hat C}\right)^\dagger\left((\cdot)\otimes\sigma^{(k)}\right)U^{(k)}_{C\to A\hat C},
	\end{equation}
	i.e. $E_k$ has the claimed form with $V_k=\left(U^{(k)}\right)^\dagger$. The form of the decryption map then follows immediately by correctness.
	\qed
\end{proof}

\section{Proof of characterization theorem}\label{app:proofs}

This section is dedicated to proving the characterization theorem for non-malleable quantum encryption schemes, i.e., \expref{Theorem}{thm:eps-effective-char}. We begin with two preparatory lemmas.

\begin{lem}\label{lem:E-of-offd}
	For any $\SKQES$ $(\tau, E,D)$ the map $\mathcal{E}:=|K|^{-1}\sum_{k}D_k\otimes E_{k}^T$ satisfies
	$$\mathcal{E}\left(\proj{\phi^+}_{CC'}(X_C\otimes \id_{C'})\right)=\frac{|A|}{|C|}\proj{\phi^+}_{AA'}(E_K^{\dagger}(X)\otimes \id_{A'})$$
\end{lem}
\begin{proof}
	 Using  \expref{Lemma}{lem:SKQES-char} we derive an expression for $E_k^T$,
	\begin{align}
	\overline{\tr\left[E_k^T(\overline Y)\overline X\right]}=&\tr\left[E_k^\dagger(Y)X\right]\nonumber\\
	=&\tr\left[Y_CE_k(X_A)\right]\nonumber\\
	=&\tr\left[Y_CV_k(X_A\otimes\sigma^{(k)}_{\hat C})V_k^\dagger\right]\nonumber\\
	=&\tr\left[\tr_{\hat C}\left(\sigma^{(k)}_{\hat C}V_k^\dagger YV_k\right)X\right]\nonumber\\
	=&\overline{\tr\left[\tr_{\hat C}\left(\overline\sigma^{(k)}_{\hat C}V_k^T \overline Y\overline V_k\right)\overline X\right]}.
	\end{align}
	Here we use the definition of the adjoint in the second equality and the cyclicity of the trace in the third equality. Hence
	\begin{equation}\label{eq:E-transpose}
	E_k^T=\tr_{\hat C}\left(\overline\sigma^{(k)}_{\hat C}V_k^T (\cdot)\overline V_k\right)
	\end{equation}
	Define $\Pi_k=\Pi_{\supp\sigma^{(k)}}$ to be the projector onto the support of $\sigma_k$. In the following we omit the subscripts of CP maps and isometries to save space. 
	We start with one summand in the sum defining $\mathcal E$ and omit the second summand from the expression for $D_k$ in \expref{Equation}{eq:Dkchar}
	\begin{align}\label{eq:summand1}
	&\tr_{\hat C\hat C'}\left[\left(\Pi_k\right)_{\hat C}\otimes\overline\sigma^{(k)}_{\hat C'}\right]\left[V_k^\dagger\otimes V_k^T\right]\phi^+_{CC'}X_C\left[V_k\otimes\overline V_k\right]\nonumber\\
	=&\frac{|A||\hat C|}{|C|}\tr_{\hat C\hat C'}\left[\left(\Pi_k\right)_{\hat C}\otimes\overline\sigma^{(k)}_{\hat C'}\right]\left[\phi^+_{AA'}\otimes \phi^+_{\hat C\hat C'}\right]V_k^\dagger X_CV_k\nonumber\\
	=&\frac{|A||\hat C|}{|C|}\tr_{\hat C\hat C'}\phi^+_{AA'}\otimes \phi^+_{\hat C\hat C'}V_k^\dagger X_CV_k\sigma^{(k)}_{\hat C}\nonumber\\
	=&\frac{|A||\hat C|}{|C|}\phi^+_{AA'} \bra{\phi^+}_{\hat C\hat C'}V_k^\dagger X_CV_k\sigma^{(k)}_{\hat C}\ket{\phi^+}_{\hat C\hat C'}\nonumber\\
	=&\frac{|A|}{|C|}\phi^+_{AA'}\tr_{\hat C}V_k^\dagger X_CV_k\sigma^{(k)}_{\hat C}\nonumber\\
	=&\frac{|A|}{|C|}\phi^+_{AA'}\left(E_k\right)_{C\to A}^{\dagger}(X_C).
	\end{align}
	Here we have used \expref{Lemma}{lem:genmirr} for the first  and the second equality, and in the fourth equality is due to the elementary fact that
	\begin{equation}
	\bra{\phi^+}_{\hat C\hat C'}Y_{\hat C}\ket{\phi^+}_{\hat C\hat C'}=\frac{1}{|\hat C|}\tr Y.
	\end{equation}
	Finally we have used the complex conjugate of \expref{Equation}{eq:E-transpose} in the last equation.
	Now we look at the same expression but only taking the second summand from \expref{Equation}{eq:Dkchar} into account.
	\begin{align}\label{eq:summand2}
	&\left\{\hat D_k\otimes \tr_{\hat C'}\right\}\left[\left(\mathds 1-\Pi_k\right)_{\hat C}\otimes\overline\sigma^{(k)}_{\hat C'}\right]\left[V_k^\dagger\otimes V_k^T\right]\phi^+_{CC'}X_C\left[V_k\otimes\overline V_k\right]\nonumber\\
	=&\frac{|A||\hat C|}{|C|}\left(\hat D_k\right)_{A\hat C\to A}\otimes \tr_{\hat C'}\left(\mathds 1-\Pi_k\right)_{\hat C}\otimes\overline\sigma^{(k)}_{\hat C'}\phi^+_{AA'}\otimes \phi^+_{\hat C\hat C'}V_k^\dagger X_CV_k\nonumber\\
	=&0
	\end{align}
	where the steps are the same as above and in the last equality we used that $\left(\mathds 1-\Pi_k\right)_{\hat C}\sigma_{\hat{C}}=0$. Adding Equations \eqref{eq:summand1} and \eqref{eq:summand2}, summing over $k$ and normalizing finishes the proof.
\end{proof}

\begin{lem}\label{lem:eps-like2des}
	Suppose $(\tau_K,E, D)$ satisfies \expref{Definition}{def:eps-qNM} for trivial $B$.
	Then $\mathcal{E}:=|K|^{-1}\sum_{k}D_k\otimes E_{k}^T$ satisfies
		\begin{align}
	&\Bigg\|\mathcal{E}(X)-\frac{|A|}{|C|}\bigg[\bra{\phi^+}X\ket{\phi^+}\proj{\phi^+}+\tr\left(\Pi^-X\right)\frac{1}{|C|^2-1}\left(|C|^2D_K(\tau_C)_A\otimes\tau_{A'}-\phi^+_{AA'}\right)\bigg]\Bigg\|_\diamond\nonumber\\
	&\le 2\sqrt{2\varepsilon}|A|\left(2\sqrt{|A|}|C|+1\right).
	\end{align}	
\end{lem}
\begin{proof}
	It follows directly from the fact that $(\tau_K,E, D)$ is a \SKQES~together with \expref{Lemma}{lem:genmirr} that
	\begin{equation}\label{eq:like2design1-in-eps-proof}
	\mathcal E(\phi^+_{CC'})=\frac{|A|}{|C|}\phi^+_{AA'}.
	\end{equation}
	
	Let $\Lambda^{(i)}_{C\to C\tilde B_1}$, $i=0,1$ be two attack maps such that $\eta_{\Lambda^{(i)}}\ket{\phi^+}=0$ for $i=0,1$ and define $$\Lambda_{C\to C\tilde B_1\tilde B_2}=\frac{1}{2}\sum_{i=0,1}\proj i_{\tilde B_2}\otimes \Lambda^{(i)}.$$ The the $\varepsilon$-$\ITNM$ property implies
	$$I(AA':\tilde B_1\tilde B_2)_{\eta_{\tilde\Lambda}}\le \varepsilon,$$
	 and therefore, using Pinsker's inequality, \expref{Lemma}{lem:pinsker},
	\begin{align}
		&\Bigg\|	\frac{1}{2}\sum_{i=0,1}\proj i_{\tilde B}\otimes \left(\eta_{\tilde\Lambda^{(i)}}\right)_{CC'\tilde B_1}\nonumber\\
		&-\frac 1 4\left(\sum_{i=0,1}\proj i_{\tilde B}\otimes \left(\eta_{\tilde\Lambda^{(i)}}\right)_{\tilde B_1}\right)\otimes\left(\sum_{i=0,1} \left(\eta_{\tilde\Lambda^{(i)}}\right)_{CC'}\right)\Bigg\|_1\le\sqrt{2\varepsilon}.
	\end{align}
	Observe that
	\begin{align}\label{eq:tildecj}
		\eta_{\tilde \Lambda}=&\frac{1}{|K|}\sum_kD_k\circ\Lambda\circ E_k(\phi^+_{AA'})\nonumber\\
		=&\frac{|C|}{|A|}\frac{1}{|K|}\sum_k\left(D_k\otimes E_k^T\right)\circ\Lambda(\phi^+_{CC'})\nonumber\\
		=&\frac{|C|}{|A|}\mathcal E\circ\Lambda(\phi^+_{CC'}).
	\end{align}
	Setting $\left(\eta_{\Lambda^{(0)}}\right)_{CC'\tilde B_1}=\tau^-_{CC'}\otimes\left(\eta_{\Lambda^{(1)}}\right)_{\tilde B_1} $, 
	we get
	\begin{align}\label{eq:etanull}
		\eta_{\tilde\Lambda^{(0)}}=&\frac{|C|}{|A|}\mathcal E(\tau^-)\otimes\left(\eta_{\Lambda^{(1)}}\right)_{\tilde B_1} \nonumber\\
		=&\frac{|C|}{|A|}\frac{1}{|C|^2-1}\left(|C|^2 \mathcal{E}(\tau_{CC'})-\mathcal{E}(\phi^+_{CC'})\right)\otimes\left(\eta_{\Lambda^{(1)}}\right)_{\tilde B_1} \nonumber\\
		=&\frac{1}{|C|^2-1}\left(|C|^2 D_K(\tau_C)\otimes\tau_A -\phi^+_{AA'}\right)\otimes\left(\eta_{\Lambda^{(1)}}\right)_{\tilde B_1} .
	\end{align}
	and therefore
	\begin{align}\label{eq:eps-like2des1}
	&\Bigg\|\frac{1}{|C|^2-1}\left(|C|^2 D_K(\tau_C)\otimes\tau_A -\phi^+_{AA'}\right)\otimes\left(\eta_{\Lambda^{(1)}}\right)_{\tilde B_1}-\frac{|C|}{|A|}\mathcal{E}\left(\left(\eta_{\Lambda^{(1)}}\right)_{CC'\tilde B_1}\right)\Bigg\|_1\le 2\sqrt{2\varepsilon}
	\end{align}
	for all $\Lambda^{(1)}$.
	For any state $\rho_{CC'\tilde B_1}$ with $\rho_{CC'\tilde B}\ket{\phi^+}_{CC'}=0$, we define the state
	$$\rho'_{CC'\tilde B_1\tilde B_2}=\frac{1}{C}\left(\proj 0_{\tilde B_2}\otimes\rho_{CC'\tilde B_1}+\proj{1}_{\tilde B_2}\otimes \left[\left((\mathds{1}_C-\rho_C)\otimes V_{C'}\right)\phi^+\left((\mathds{1}_C-\rho_C)\otimes V_{C'}\right)\right]\otimes\rho_{\tilde B_2}\right).$$
	Here, $V$ is a unitary such that $\tr(\mathds{1}_C-\rho_C)V_C^T=0$. It is easy to see that such a unitary always exists, the existence is equivalent to the fact that any $|C|$-tuple of real numbers is the ordered list of side lengths of a polygon in the complex plain. Note that $\rho'_{CC'\tilde B_1\tilde B_2}\ket{\phi^+}_{CC'}=0$, and $\rho'_{C'}=\tau_{C'}$.
	Together with the triangle inequality, equation \eqref{eq:eps-like2des1} implies therefore that
	\begin{eqnarray}
		\frac{1}{|C|}&\bigg\|&\frac{|C|}{|A|}\mathcal{E}(\rho)-\frac{1}{|C|^2-1}\left(|C|^2 D_K(\tau_C)\otimes\tau_A -\phi^+_{AA'}\right)\otimes\rho_{\tilde B_1}\bigg\|_1\nonumber\\
		+&\bigg\|&\frac{|C|}{|A|}\mathcal{E}\left[\left((\mathds{1}_C-\rho_C)\otimes V_{C'}\right)\phi^+\left((\mathds{1}_C-\rho_C)\otimes V_{C'}\right)\right]\nonumber\\
		&&-\frac{|C|-1}{|C|}\frac{1}{|C|^2-1}\left(|C|^2 D_K(\tau_C)\otimes\tau_A -\phi^+_{AA'}\right)\bigg\|_1\le 2\sqrt{2\varepsilon},\nonumber
	\end{eqnarray}
	i.e. in particular
	\begin{equation}
		\bigg\|\frac{|C|}{|A|}\mathcal{E}(\rho)-\frac{1}{|C|^2-1}\left(|C|^2 D_K(\tau_C)\otimes\tau_A -\phi^+_{AA'}\right)\otimes\rho_{\tilde B_1}\bigg\|_1\le 2\sqrt{2\varepsilon}|C|.\nonumber
	\end{equation}
	As $\rho$ was arbitrary we have proven that
	\begin{equation}\label{eq:diamond-E}
	\bigg\|\frac{|C|}{|A|}\mathcal{E}-\left\langle\frac{1}{|C|^2-1}\left(|C|^2 D_K(\tau_C)\otimes\tau_A -\phi^+_{AA'}\right)\right\rangle\bigg\|_\diamond\le 2\sqrt{2\varepsilon}|C|.
	\end{equation}

	The only fact that is left to show is, that $\|\mathcal{E}(\ketbra{\phi^+}{v})\|_1$ is small for all normalized $\ket{v}$ such that $\braket{\phi^+}{v}=0$.
	To this end, observe that $\tr_{A}\circ \mathcal E(\sigma_C\otimes(\cdot)_{C'})=E_K^T$ for all  quantum states $\sigma_C$. Let $\rho_C$ be any quantum state that does not have full rank, note that such states span all of $\opr(\hi_C)$, and for hermitian operators there exists a decomposition into such operators that saturates the triangle inequality. Taking a quantum state $\sigma_C$ such that $\bra{\phi^+}\rho\otimes\sigma\ket{\phi^+}=\frac{1}{|C|}\tr\rho_C\sigma_C^T=0$ (the first equality is the mirror lemma \ref{lem:genmirr}), we have 
	$$\left\|\mathcal{E}(\rho\otimes\sigma)-\frac{|A|}{|C|}\frac{1}{|C|^2-1}\left(|C|^2 D_K(\tau_C)\otimes\tau_A -\phi^+_{AA'}\right)\right\|_1\le 2\sqrt{2\varepsilon}|A|$$
	 according to what we have already proven. 
		Using inequality \eqref{eq:diamond-E} we arrive at
	\begin{equation}
	\left\|E_K^\dagger(X)-\frac{|A|}{|C|}\tau_A\tr(X)\right\|_{1}\le 2\sqrt{2\varepsilon}|A|\|X\|_1
	\end{equation}
	For Hermitian matrices $X$ and therefore
	\begin{equation}\label{eq:EK-bound}
	\left\|E_K^\dagger(X)-\frac{|A|}{|C|}\tau_A\tr(X)\right\|_{1}\le 4\sqrt{2\varepsilon}|A|\|X\|_1
	\end{equation}
	For arbitrary $X$.
	 We can write $\ket{v}_{CC'}=X_C\ket{\phi^+}_{CC'}$ for some traceless matrix $X_C$. Now we calculate
	\begin{align}
		\left\|\mathcal E(\ketbra{\phi^+}{v}_{CC'})\right\|_1=&\left\|\frac{|A|}{|C|}\proj{\phi^+}_{AA'}\left(E^\dagger_K(X^\dagger)\right)_A\right\|_1\nonumber\\
		=&\frac{|A|}{|C|}\left\|\left(E^\dagger_K(X)\right)_A\ket{\phi^+}_{AA'}\right\|_2\nonumber\\
		=&\frac{\sqrt{|A|}}{|C|}\left\|E^\dagger_K(X)\right\|_2\nonumber\\
		\le&\frac{\sqrt{|A|}}{|C|}\left\|E^\dagger_K(X)\right\|_1\nonumber\\
		\le&\frac{|A|^{3/2}}{|C|}4\sqrt{2\varepsilon}\|X\|_1\nonumber\\
		\le&4\sqrt{2\varepsilon}|A|^{3/2}.
	\end{align}
	The first equation is \expref{Lemma}{lem:E-of-offd}, the sencond and third equations are easily verified, the first inequality is a standard norm inequality, the second inequality is Equation \eqref{eq:EK-bound}, and the last inequality follows from the normalization of $\ket v$.
	By the Schmidt decomposition, we get a stabilized version of this inequality,
	\begin{align}
	\left\|\mathcal E(\ket{\phi^+}_{CC'}\ket{\alpha}_{\tilde B_1}\bra{v}_{CC'\tilde B_1})\right\|_1\le&2\sqrt{2\varepsilon}|A|^{3/2},
	\end{align}
	for all $\ket\alpha_{\tilde B_1}$ and all $\ket{v}_{CC'\tilde B}$ such that $\braket{\phi^+}{v}=0$
	Combining everything we arrive at
	\begin{align}
	&\Bigg\|\mathcal{E}(X)-\frac{|A|}{|C|}\bigg[\bra{\phi^+}X\ket{\phi^+}\proj{\phi^+}\nonumber\\
	&+\tr\left(\Pi^-X\right)\frac{1}{|C|^2-1}\left(|C|^2D_K(\tau_C)_A\otimes\tau_{A'}-\phi^+_{AA'}\right)\bigg]\Bigg\|_\diamond\le 2\sqrt{2\varepsilon}|A|\left(4\sqrt{|A|}+1\right).
	\end{align}	
	\qed
\end{proof}

We are now ready to prove the characterization theorem (i.e., \expref{Theorem}{thm:eps-effective-char}) in the $\varepsilon$-approximate setting. We remark that the exact setting (stated and sketched as \expref{Theorem}{thm:effective-char} in the main text) is simply the case where $\varepsilon = 0$.
\begin{thm}[Precise version of \expref{Theorem}{thm:eps-effective-char}]\label{thm:eps-effective-char-app}
	Let $\Pi=(\tau, E,D)$ be a \SKQES.
	\begin{enumerate}
		\item If $\Pi$ is $\varepsilon$-\ITNM, then any attack map $\Lambda_{CB\to C\tilde B}$ results in an effective map $\tilde\Lambda_{AB \to A{\tilde B}}$ fulfilling
		\begin{equation}\label{eq:eps-effective-map}
			\left\|\tilde\Lambda_{AB \to A{\tilde B}}-\tilde \Lambda^{\mathrm{exact}}_{AB\to A\tilde B}\right\|_\diamond\le 2\sqrt{2\varepsilon}|A|^4|C|\left(4\sqrt{|A|}+1\right),
		\end{equation}
		where
		\begin{equation*}
		\tilde \Lambda^{\mathrm{exact}}_{AB\to A\tilde B}=\id_A\otimes \Lambda'_{B\to\tilde B}+\frac{1}{|C|^2-1}\left(|C|^2\left\langle D_K(\tau)\right\rangle-\id\right)_A\otimes \Lambda''_{B\to\tilde B},
		\end{equation*}
		with $\Lambda' =\tr_{CC'}[\phi^+_{CC'}\Lambda(\phi^+_{CC'}\otimes (\cdot))]$ and $\Lambda'' =\tr_{CC'}[\Pi^-_{CC'}\Lambda(\phi^+_{CC'}\otimes (\cdot))].$
		\item Conversely, if for a scheme all effective maps fulfil Equation \eqref{eq:eps-effective-map} with the right hand side replaced by $\varepsilon$, then it is $5\varepsilon(\log(|A|)+r)+3h(\varepsilon)$-\ITNM, where $r$ is a bound on the size of the honest user's side information.
	\end{enumerate}
\end{thm}
\begin{proof}
	We start with \textit{1.} We want to bound the diamond norm distance between the effective map $\tilde\Lambda$ resulting from an attack $\Lambda$ and the idealized effective map $\tilde\Lambda^{\mathrm{exact}}$. Let 
	$$\ket\psi_{AA'BB'}=\sum_{i=0}^{|A|^2-1}\sqrt{p_i}\ket{\alpha_i}_{AA'}\otimes\ket{\beta_i}_{BB'}$$
	 be an arbitrary pure state given in its Schmidt decomposition across the bipartition $AA'$ vs. $BB'$. We can Write $\ket{\alpha_i}_{AA'}=X^{(i)}_{A'}\ket{\phi^+}$ for some matrices $X^{(i)}$ satisfying $\|X^{(i)}\|_\infty\le|A|$. We calculate the action of $\tilde \Lambda$ on $\ketbra{\alpha_i}{\alpha_j}_{AA'}\otimes\ketbra{\beta_i}{\beta_j}_{BB'}$,
	 \begin{align}
	 &\tilde\Lambda^{\mathrm{exact}}_{AB \to A{\tilde B}}(\ketbra{\alpha_i}{\alpha_j}_{AA'}\otimes\ketbra{\beta_i}{\beta_j}_{BB'})=X^{(i)}_{A'}\bigg(\proj{\phi^+}_{AA'}\otimes\Lambda'_{B\to \tilde B}(\ketbra{\beta_i}{\beta_j}_{BB'})\nonumber\\
	 &+\frac{1}{|C|^2-1}\left(|C|^2 D_K(\tau)_A\otimes\tau_{A'}-\proj{\phi^+}_{AA'}\right)\otimes\Lambda''_{B\to \tilde B}(\ketbra{\beta_i}{\beta_j}_{BB'})\bigg)X^{(j)}_{A'}.
	 \end{align}
	In a similar way we get
	\begin{align}
	\tilde\Lambda_{AB \to A{\tilde B}}(\ketbra{\alpha_i}{\alpha_j}_{AA'}\otimes\ketbra{\beta_i}{\beta_j}_{BB'})&=X^{(i)}_{A'}\tilde\Lambda_{AB \to A{\tilde B}}(\proj{\phi^+}_{AA'}\otimes\ketbra{\beta_i}{\beta_j}_{BB'})X^{(i)}_{A'}\nonumber\\
	&=\frac{|C|}{|A|}X^{(i)}_{A'}\mathcal{E}_{CC'\to AA'}\circ\Lambda_{CB \to C{\tilde B}}(\proj{\phi^+}_{CC'}\otimes\ketbra{\beta_i}{\beta_j}_{BB'})X^{(i)}_{A'}.
	\end{align}
	Using \expref{Lemma}{lem:eps-like2des} we bound
	\begin{align}
		&\left\|\left(\tilde\Lambda_{AB \to A{\tilde B}}-\tilde\Lambda^{\mathrm{exact}}_{AB \to A{\tilde B}}\right)(\ketbra{\alpha_i}{\alpha_j}_{AA'}\otimes\ketbra{\beta_i}{\beta_j}_{BB'})\right\|_1\nonumber\\
		=&\left\|X^{(i)}_{A'}\left(\tilde\Lambda_{AB \to A{\tilde B}}-\tilde\Lambda^{\mathrm{exact}}_{AB \to A{\tilde B}}\right)(\proj{\phi^+}_{AA'}\otimes\ketbra{\beta_i}{\beta_j}_{BB'})X^{(j)}\right\|_1\nonumber\\
		\le&\left\|X^{(i)}\right\|_\infty\left\|X^{(j)}\right\|_\infty\left\|\left(\tilde\Lambda_{AB \to A{\tilde B}}-\tilde\Lambda^{\mathrm{exact}}_{AB \to A{\tilde B}}\right)(\proj{\phi^+}_{AA'}\otimes\ketbra{\beta_i}{\beta_j}_{BB'})\right\|_1\nonumber\\
		=&\left\|X^{(i)}\right\|_\infty\left\|X^{(j)}\right\|_\infty\bigg\|\frac{|C|}{|A|}\mathcal E_{CC'\to AA'}\circ\Lambda_{CB \to C{\tilde B}}(\proj{\phi^+}_{CC'}\otimes\ketbra{\beta_i}{\beta_j}_{BB'})\nonumber\\
		&~~~~~~~~~~~~~~~~~~~~~~~~~~~-\tilde\Lambda^{\mathrm{exact}}_{AB \to A{\tilde B}}(\proj{\phi^+}_{AA'}\otimes\ketbra{\beta_i}{\beta_j}_{BB'})\bigg\|_1\nonumber\\
		\le&2\sqrt{2\varepsilon}|A|^2|C|\left(4\sqrt{|A|}+1\right).
	\end{align}
	The inequalities result from applying H\"older's inequality twice, and  \expref{Lemma}{lem:eps-like2des}, respectively.	Using the triangle inequality we get
	\begin{align}
		\left\|\left(\tilde\Lambda_{AB \to A{\tilde B}}-\tilde\Lambda^{\mathrm{exact}}_{AB \to A{\tilde B}}\right)(\proj{\psi}_{AA'BB'})\right\|_1\le&2\sqrt{2\varepsilon}|A|^2|C|\left(4\sqrt{|A|}+1\right)\sum_{i,j=0}^{|A|^2-1}\sqrt{p_ip_j}\nonumber\\
		\le&2\sqrt{2\varepsilon}|A|^4|C|\left(4\sqrt{|A|}+1\right).
	\end{align}
	As $\ket{\psi}$ was arbitrary, we have proven
	\begin{align}
	\left\|\tilde\Lambda_{AB \to A{\tilde B}}-\tilde\Lambda^{\mathrm{exact}}_{AB \to A{\tilde B}}\right\|_\diamond\le&2\sqrt{2\varepsilon}|A|^4|C|\left(4\sqrt{|A|}+1\right).
	\end{align}

	Now let us prove $2.$ Let $\Lambda_{CB\to C\tilde B}$ again be an arbitrary attack map, and assume that the resulting effective map is $\varepsilon$-close to $\tilde\Lambda^{\mathrm{exact}}_{AB\to A\tilde B}$. Observe that $p^{=}(\Lambda,\rho)=\tr\Lambda'(\rho_B)$.
	
	By \expref{Lemma}{lem:fannes} and \expref{Lemma}{lem:DP-CPTPtensCP}, this implies
	\begin{equation}
	I(AR:\tilde B)_{\tilde\Lambda(\rho)}\le I(AR:B)_\rho+h(p^=(\Lambda,\rho))+5\varepsilon\log(|A||R|)+3h(\varepsilon)
	\end{equation}
	with the help of \expref{Lemma}{lem:DP-CPTPtensCP} (given in the appendix of the main article).
	\qed
\end{proof}

\end{document}